\documentclass[a4paper,thm-restate,cleveref, autoref]{lipics-v2021}
\usepackage[utf8]{inputenc}
\usepackage{amsmath,amsfonts,amssymb,amsthm,enumerate}
\usepackage{bbm}
\usepackage{mathtools}
\usepackage{url}
\usepackage{todonotes}
\usepackage{hyperref}
\usepackage{comment}
\usepackage{algorithm}
\usepackage{algpseudocode}
\usepackage{xspace}
\title{Fast Approximations and Coresets for $(k,\ell)$-Median under Dynamic Time Warping}

\author{Jacobus Conradi}{University of Bonn, Bonn, Germany}{jacobus.conradi@gmx.de}{https://orcid.org/0000-0002-8259-1187}{Partially funded by the Deutsche Forschungsgemeinschaft (DFG, German Research Foundation) - 313421352 (FOR 2535 Anticipating Human Behavior) and the iBehave Network: Sponsored by the Ministry of Culture and Science of the State of North Rhine-Westphalia. Affiliated with Lamarr Institute for Machine Learning and Artificial Intelligence.}

\author{Benedikt Kolbe}{Hausdorff Center for Mathematics, University of Bonn, Bonn, Germany}{benedikt.kolbe@physik.hu-berlin.de}{}{Partially funded by the Deutsche Forschungsgemeinschaft (DFG, German Research Foundation) – 459420781.}

\author{Ioannis Psarros}{Athena Research Center, Marousi,  Greece}{ipsarros@athenarc.gr}{https://orcid.org/0000-0002-5079-5003}{Supported by project MIS 5154714 of the National Recovery and Resilience Plan Greece 2.0 funded by the European Union under the NextGenerationEU Program. }

\author{Dennis Rohde}{University of Bonn, Bonn, Germany}{}{https://orcid.org/0000-0001-8984-1962}{Partially supported by the Hausdorff Center for Mathematics. Partially funded by the iBehave Network: Sponsored by the Ministry of Culture and Science of the State of North Rhine-Westphalia.}

\authorrunning{Conradi, Kolbe, Psarros and Rohde} 

\Copyright{Jacobus Conradi, Benedikt Kolbe, Ioannis Psarros and Dennis Rohde} 



\category{} 

\relatedversion{} 

\acknowledgements{We thank Anne Driemel of the University of Bonn for detailed discussion and guidance.}

\EventEditors{John Q. Open and Joan R. Access}
\EventNoEds{2}
\EventLongTitle{42nd Conference on Very Important Topics (CVIT 2016)}
\EventShortTitle{CVIT 2016}
\EventAcronym{CVIT}
\EventYear{2016}
\EventDate{December 24--27, 2016}
\EventLocation{Little Whinging, United Kingdom}
\EventLogo{}
\SeriesVolume{42}
\ArticleNo{23}

\DeclarePairedDelimiter{\parens}{\lparen}{\rparen}

\def\NN{{\mathbb N}}
\def\RR{{\mathbb R}}
\def\bN{{\mathbb N}}
\def\bR{{\mathbb R}}
\def\bX{{\mathbb X}}
\def\RRR{{\mathcal R}}
\def\dist{\phi}
\def\distconst{\zeta}
\DeclareMathOperator{\sign}{\mathrm{sign}}
\DeclareMathOperator{\dtw}{dtw}
\DeclareMathOperator{\dtwp}{{dtw_{p}}}
\DeclareMathOperator{\adtwp}{\widetilde{\dtwp}}

\DeclareMathOperator{\pcost}{cost}
\DeclareMathOperator{\pball}{B}

\DeclareMathOperator{\pexpected}{E}
\DeclareMathOperator{\prange}{range}
\DeclareMathOperator{\power}{\mathcal{P}}

\newcommand{\ind}[1]{\ensuremath{\mathbbm{1}\parens*{#1}}}
\newcommand{\scost}[2]{\ensuremath{\widehat{\pcost}\parens*{#1,#2}}}
\newcommand{\ucost}[2]{\ensuremath{\overline{\pcost}\parens*{#1,#2}}}
\newcommand{\acost}[2]{\ensuremath{\widetilde{\pcost}\parens*{#1,#2}}}
\newcommand{\cost}[2]{\ensuremath{\pcost\parens*{#1,#2}}}
\newcommand{\range}[3]{\ensuremath{\prange\parens*{#1,#2,#3}}}
\newcommand{\expected}[1]{\ensuremath{\pexpected\left[#1\right]}}

\newcommand{\opt}{^\mathrm{opt}}

\newcommand{\traversals}{\ensuremath{\mathcal{T}}}

\let\epsilon\relax
\newcommand{\epsilon}{\varepsilon}

\newif\ifstartedinmathmode
\newcommand{\ben}[1]{%
\relax\ifmmode\startedinmathmodetrue\else\startedinmathmodefalse\fi
\ifstartedinmathmode\textbf{\textcolor{cyan}{[Ben: #1]}}\else\par\noindent\ignorespaces\begin{nolinenumbers}\textbf{\textcolor{cyan}{[Ben: #1]}}\end{nolinenumbers}\fi%
}
\newcommand{\jacobus}[1]{%
\relax\ifmmode\startedinmathmodetrue\else\startedinmathmodefalse\fi
\ifstartedinmathmode\textbf{\textcolor{blue}{[Jacobus: #1]}}\else\par\noindent\ignorespaces\begin{nolinenumbers}\textbf{\textcolor{blue}{[Jacobus: #1]}}\end{nolinenumbers}\fi%
}
\newcommand{\dennis}[1]{%
\relax\ifmmode\startedinmathmodetrue\else\startedinmathmodefalse\fi
\ifstartedinmathmode\textbf{\textcolor{green}{[Dennis: #1]}}\else\par\noindent\ignorespaces\begin{nolinenumbers}\textbf{\textcolor{green}{[Dennis: #1]}}\end{nolinenumbers}\fi%
}
\newcommand{\jannis}[1]{%
\relax\ifmmode\startedinmathmodetrue\else\startedinmathmodefalse\fi
\ifstartedinmathmode\textbf{\textcolor{red}{[Jannis: #1]}}\else\par\noindent\ignorespaces\begin{nolinenumbers}\textbf{\textcolor{red}{[Jannis: #1]}}\end{nolinenumbers}\fi%
}

\def\eps{\varepsilon}
\def\bR{{\mathbb R}}
\def\bN{{\mathbb N}}
\def\bZ{{\mathbb Z}}
\def\bX{{\mathbb X}}

\newcommand{\curvespace}[1]{\bX^d_{#1}}

\begin{document}
\nolinenumbers
\hideLIPIcs

\maketitle
\begin{abstract}
We present algorithms for the computation of $\epsilon$-coresets for $k$-median clustering of point sequences in $\mathbb{R}^d$ under the $p$-dynamic time warping (DTW) distance. Coresets under DTW have not been investigated before, and the analysis is not directly accessible to existing methods as DTW is not a metric. The three main ingredients that allow our construction of coresets are the adaptation of the $\epsilon$-coreset framework of sensitivity sampling, bounds on the VC dimension of approximations to the range spaces of balls under DTW, and new approximation algorithms for the $k$-median problem under DTW. We achieve our results by investigating approximations of DTW that provide a trade-off between the provided accuracy and amenability to known techniques. In particular, we observe that given $n$ curves under DTW, one can directly construct a metric that approximates DTW on this set, permitting the use of the wealth of results on metric spaces for clustering purposes. The resulting approximations are the first with polynomial running time and achieve a very similar approximation factor as state-of-the-art techniques.
We apply our results to produce a practical algorithm approximating $(k,\ell)$-median clustering under DTW.	

\end{abstract}

\section{Introduction}
One of the core challenges of contemporary data analysis is the handling of massive data sets. 
A powerful approach to clustering problems involving such sets is data reduction, and $\epsilon$-coresets offer a popular approach that has received substantial attention. An $\epsilon$-coreset is a problem-specific condensate of the given input set of reduced size which captures its core properties towards the problem at hand and can be used as a proxy to run an algorithm on, producing a solution with a relative error of $(1\pm\epsilon)$. 

Clustering and especially $k$-median represent fundamental tasks in classification problems, where they have been extensively studied for various spaces. With the growing availability of e.g. geospatial tracking data, clustering problems for time series or curves have received growing attention both from a theoretical and applied perspective. In practice, time series classification largely relies on the dynamic time warping (DTW) distance and is widely used in the area of data mining. Simple nearest neighbor classifiers under DTW are considered hard to beat \cite{DBLP:journals/datamine/Kate16,DBLP:journals/eaai/TranLNH19} and much effort has been put into making classification using DTW computationally efficient~\cite{DBLP:journals/pr/JeongJO11,DBLP:conf/icdm/PetitjeanFWNCK14,DBLP:journals/kais/PetitjeanFWNCK16,DBLP:journals/tkdd/RakthanmanonCMBW0ZK13}. In contrast to its cousin the Fréchet distance, DTW 
is less sensitive to outliers, but its algorithmic properties are also less well understood, owing to the fact that it is not a metric. In particular, the wealth of research surrounding $k$-median clustering for metric spaces does not directly apply to clustering problems under DTW.  

For time series and curves, $k$-median takes the shape of the $(k,\ell)$-median problem, where the sought-for center curves are restricted to have a complexity (number of vertices) of at most $\ell$, with a two-fold motivation. First, the otherwise NP-hard problem becomes tractable, and second, it suppresses overfitting.  

The construction of $\epsilon$-coresets for the $(k,\ell)$-median problem for DTW is precisely what this paper will address. To this end, we adapt the framework of \emph{sensitivity sampling} by Feldman and Landberg~\cite{Feldman2011} to our setting, derive bounds on the VC dimension of approximate range spaces of balls under DTW, 
develop fast approximation algorithms solving $(k,\ell)$-median clustering, and use coresets to improve existing $(k,\ell)$-median algorithms, for curves under DTW. 
We rely on approximations of nearly all objects involved in our inquiry, thereby improving the bounds we obtain for the VC dimension of the range spaces in question and broadening the scope of our approach.  

Our analysis of the VC dimension is possibly of independent interest. 
The VC dimension exhibits a near-linear dependency on the complexity of the sequences used as centers of the ranges, yet it depends only logarithmically on the size of the curves within the ground set. This distinction holds significant implications in the analysis of real datasets, where queries may involve simple, short sequences, but the dataset itself may consist of complex, lengthy sequences. 
Note that our results hold for range spaces that are defined by small perturbations of DTW distances. 
This means that for any given set of input sequences requiring DTW-based analysis, there is slight perturbation of DTW with associated range space of bounded VC dimension.
This is sufficient to enable a broad array of algorithmic techniques that leverage the VC dimension, particularly in scenarios where approximate computations are allowed.

\subparagraph{Related Work}
Among different practical approaches for solving the $k$-median problem, a very influential heuristic is the DTW Barycentric Average (DBA) method~\cite{PKG11}. While it has seen much success in practice~\cite{ACS03,HNF08,RW79}, it does not have any theoretical guarantees and indeed may converge to a local configuration that is arbitrarily far from the optimum. Recently, theoretical results for average series problems under DTW have been obtained. The problem is $\mathrm{NP}$-hard for rational-valued time series and $\mathrm{W}[1]$-hard in the number of input time series~\cite{BDS20,BFN20}. Furthermore, it can not be solved in time $O(f(n))\cdot m^{o(n)}$ for any computable function $f$ unless the Exponential Time Hypothesis (ETH) fails. There is an exponential time exact algorithm for rational-valued time series \cite{DBLP:journals/datamine/BrillFFJNS19} and polynomial time exact algorithms for binary time series \cite{DBLP:journals/datamine/BrillFFJNS19,SFN20}. There is an exact algorithm for the related problem of finding a single mean curve of given complexity for time series over $\mathbb{Q}^d$, minimizing the sum of squares of DTW distances to input curves, which runs in polynomial time if the number of points of the average series is constant~\cite{DBLP:conf/waoa/BuchinDGPR22}. Furthermore, approximation algorithms were recently developed~\cite{DBLP:conf/waoa/BuchinDGPR22}, and some of these can be slightly modified to work within the median clustering approximation framework of~\cite{doi:10.1137/1.9781611976465.160,DBLP:journals/talg/BuchinDR23}. Unfortunately, known median clustering approximation algorithms either have running time exponential in the length of the average series, or a very large approximation factor. 

\subparagraph{Approximation Algorithms for Series Clustering}
In the last decade, the problems of $(k,\ell)$-median and $(k,\ell)$-center clustering for time series in $\mathbb{R}^d$ under the Fréchet distance have gained significant attention. The problem is $\mathrm{NP}$-hard \cite{DBLP:conf/soda/BuchinDGHKLS19,DBLP:conf/swat/BuchinDS20,DBLP:conf/soda/DriemelKS16}, even if $k=1$ and $d=1$ (in these works, time series are real-valued sequences), and the $(k,\ell)$-center problem is even $\mathrm{NP}$-hard to approximate within a factor of $(2.25 - \epsilon)$ for $d\ge 2$~\cite{DBLP:conf/soda/BuchinDGHKLS19} ($(1.5 - \epsilon)$, if $d=1$). 

For the $(k,\ell)$-median problem, all presently known $(1+\epsilon)$-approximation algorithms are based on an approximation scheme~\cite{DBLP:conf/soda/BuchinDR21,DBLP:conf/soda/ChengH23,DBLP:conf/soda/DriemelKS16} which has been generalized several times~\cite{DBLP:journals/talg/AckermannBS10,DBLP:journals/talg/BuchinDR23,DBLP:conf/focs/KumarSS04}. The most recent version of this scheme~\cite[Theorem~7.2]{DBLP:journals/talg/BuchinDR23} can be utilized to approximate any $k$-median type problem in an arbitrary space $X$ with a distance function. All that it needs is a plugin-algorithm that, when given a set $T$ of elements from some (problem-specific) subset $Y \subseteq X$, returns a set of candidates $C$ that contains, for any set $T^\prime \subseteq T$ with roughly $\lvert T^\prime \rvert \geq \lvert T \rvert / k$, with a previously fixed probability, an approximate median. The resulting approximation quality and running time depend on the approximation factor of the plugin and $\lvert C \rvert$, respectively, with a factor of $O(|C|^k)$ in the running time.  


For the Fréchet distance, plugin-algorithms exist that yield $(1+\epsilon)$-approximations~\cite{DBLP:conf/soda/BuchinDR21,DBLP:conf/soda/ChengH23}. For DTW however, the best plugin-algorithm~\cite{DBLP:conf/waoa/BuchinDGPR22} has runing time exponential in $k$---roughly with a dependency of $\widetilde{O}((32k^2\eps^{-1})^{k+2}n)$---
and approximation guarantee of $(8+\eps)(m\ell)^{1/p}$ with constant success probability. 
Here, the $\widetilde{O}$ notation hides polylogarithmic factors. 
In principle, some of the ideas from plugins for the Fréchet distance could be adapted, but the more involved plugins, i.e., the ones yielding $(1+\epsilon)$-approximations, crucially make use of the metric properties of the distance function.

In practice, an adaption of Gonzalez algorithm for $(k,\ell)$-center clustering under the Fréchet distance performs well~\cite{DBLP:conf/gis/BuchinDLN19}. Similar ideas have also been used for clustering under (a continuous variant of) DTW~\cite{Brankovic2020}, but there are no approximation guarantees, and the usual analysis is based on repetitive use of the triangle inequality. To the best of our knowledge, all $(k,\ell)$-median $(1+\epsilon)$-approximation algorithms for Fréchet and DTW are impractical due to large constants and an exponential dependency on $\ell$ in the running time. 

For the Fréchet distance, $\epsilon$-coresets can be constructed~\cite{DBLP:conf/focs/BravermanCJKST022,DBLP:conf/caldam/BuchinR22} that help facilitate the practicability of available algorithms. Using $\epsilon$-coresets, a $(5+\epsilon)$-approximation algorithm for the $1$-median problem was recently analyzed~\cite{DBLP:conf/caldam/BuchinR22}, yielding a total running time of roughly $nm^{O(1)} + (m/\epsilon)^{O(\ell)}$, in contrast to a running time of $n(m/\epsilon)^{O(\ell)}$ without the use of coresets~\cite{DBLP:journals/talg/BuchinDR23}. 

For DTW, no coreset construction is known to this point. This is at least partially due to prominent coreset frameworks assuming a normed or at least a metric space~\cite{Feldman2011,LangbergS10}. Also, recently a coreset construction relying solely on uniform sampling was developed that greatly simplifies existing coreset constructions~\cite{BCJKST022}, including the aforementioned coresets under the Fréchet distance. Unfortunately, the construction again relies on different incarnations of the triangle inequality, limiting its use for DTW.

\subparagraph{Results}
To construct $\epsilon$-coresets, we use approximations of the range space defined by balls under $p$-DTW and bound their VC dimension. 
Assuming that the input is a set of $n$ curves of complexity at most $m$, we present an approximation algorithm (\Cref{thm:mainapplication}) for $k$-median with running time in $O(n)$ (hiding other factors), that improves upon existing work in terms of running time, with comparable approximation guarantees. Our approach relies on curve simplifications and approximating $p$-DTW by a path metric. This allows us to apply state-of-the-art $k$-median techniques in this nearby path metric space, circumventing the use of heavy $k$-median machinery in non-metric spaces which would incur exponential dependence on $k$ and the success probability. Our main ingredient is a new insight into the notion of relaxed triangle inequalities for $p$-DTW (\Cref{lem:quadineq}). We then construct a coreset based on the approximation algorithm. For this, we bound the so-called sensitivity of the elements of the given data set, as well as their sum. The sensitivities are a measure of the data elements' importance and determine the sample probabilities in the coreset construction. We construct an $\eps$-coreset for $(k,\ell)$-median clustering of size quadratic in $1/\eps$ and $k$, logarithmic in $n$, and depending on $(m\ell)^{1/p}$ and $\ell$ (\Cref{cor:maincoreset}). We achieve this by upper bounding the VC dimension of the approximate range space with logarithmic dependence on $m$ (\cref{thm:approxballsvcdim}).

\section{Preliminaries}\label{sec:prelims}
We think of a sequence $(p_1, \dots, p_m) \in \left(\RR^d\right)^m$ of points in $\RR^d$ as a (polygonal) \emph{curve}, with complexity $m$. We denote by $\curvespace{=m}$  the space of curves in $\RR^d$ with complexity exactly $m$ and by $\curvespace{m}$ the space of curves with complexity at most $m$.
 \begin{figure}
    \centering
    \includegraphics[width=\textwidth]{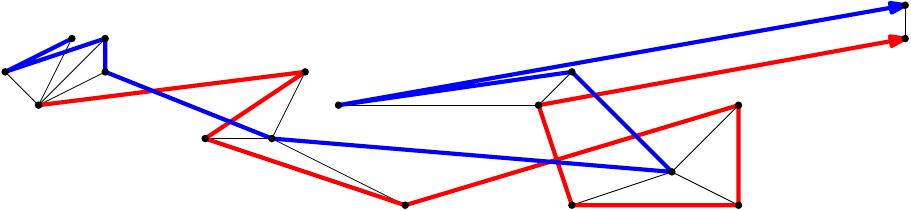}
    \caption{Example of a traversal between the red and blue curve realizing the dynamic time warping distance. The sum of the black distances is minimized.}
    \label{fig:dtwp}
\end{figure}
\begin{definition}[$p$-Dynamic Time Warping]
    \label{def:pq_dtw}

For given $m,\ell>0$ we define the space $\traversals_{m,\ell}$ of $(m,\ell)$-\emph{traversals} as the set of sequences $((a_1,b_1),(a_2,b_2),\ldots,(a_l,b_l))$, such that 
\begin{itemize}
    \item $a_1 = 1$ and $b_1 = 1$; and $a_l = m$ and $b_l = \ell$,
    \item for all $i\in [l-1]:=\{1,\ldots,l-1\}$ it holds that $(a_{i+1},b_{i+1})-(a_i,b_i)\in\{(1,0),(0,1),(1,1)\}$.
\end{itemize}
    For $p \in [1, \infty)$ and two curves $\sigma = (\sigma_1, \dots, \sigma_{m}) \in \curvespace{=m}, \tau = (\tau_1, \dots, \tau_{\ell}) \in \curvespace{=\ell}$ the \mbox{($p$-)Dynamic} Time Warping distance ($p$-DTW) is defined as $$ \dtwp(\sigma, \tau) = \min_{T \in \traversals_{m,\ell}} \left( \sum_{(i,j) \in T} \|\sigma_i- \tau_j\|_2^p \right)^{1/p} . $$ We say $\left(\sum_{(i,j)\in T}\|\sigma_i-\tau_j\|_2^p\right)^{1/p}$ is the induced cost of $T$. 
\end{definition}

The central focus of the paper is the following clustering problem.


\label{section:dtwsensitivity}
\begin{definition}[Problem definition]
    The $(k,\ell)$-median problem for $\curvespace{m}$ and $k\in\bN$ is the following: Given a set of $n\in\bN$ input curves $T=\{\tau_1,\ldots,\tau_n\}\subset \curvespace{m}$, identify $k$ center curves $C=\{c_1,\ldots,c_k\}\subset \curvespace{\ell}$ 
    that minimize $\cost{T}{C}=\sum_{\tau\in T}\min_{c\in C}\dtw(\tau,c)$.
\end{definition}
An influential approach to solving $k$-median problems is to construct a point set that acts as proxy on which to run computationally more expensive algorithms that yield solutions with approximation guarantees.  The condensed input set is known as a coreset. 
\begin{definition}[$\eps$-coreset]\label{def:epscoreset}
Let $T\subset \curvespace{m}$ be a finite set and $\epsilon\in(0,1)$. Then a weighted multiset $S\subset \curvespace{m}$ with weight function $w:S\to\mathbb{R}_{>0}$ is a weighted $\epsilon$-coreset for $(k,\ell)$-median clustering of $T$ under $\dtwp$ if for all $C\subset \curvespace{\ell}$ with $|C|=k$
\[
(1-\epsilon)\cost{T}{C}\le \sum_{s\in S}w(s) \min_{c\in C}\dtwp(s,c)\le (1+\epsilon)\cost{T}{C}.
\]
\end{definition}

\begin{definition}[$(\alpha,\beta)$-approximation]
    Let a set of $n\in\bN$ input curves $T=\{\tau_1,\ldots,\tau_n\}\subset \curvespace{m}$ be given. A set $\hat{C}\subset\curvespace{\ell}$ is called an $(\alpha,\beta)$-approximation of  $(k,\ell)$-median, if $|\hat{C}|\leq\beta k$ and $\sum_{\tau\in T}\min_{c\in \hat{C}}\dtwp(\tau,c)\leq \alpha\sum_{\tau\in T}\min_{c\in C}\dtwp(\tau,c)$ for any $C\subset\curvespace{\ell}$ of size $k$.
\end{definition}

Relaxing the problem to $(\alpha,\beta)$-approximations allows us to pass through so called simplifications of the input curves.

\begin{definition}[$(1+\eps)$-approximate $\ell$-simplifications]
 Let $\sigma\in \curvespace{m}$, $\ell\in \NN$ and $\eps>0$. We call $\sigma^*\in\curvespace{\ell}$ an \emph{$(1+\eps)$-approximate $\ell$-simplification} if 
    $$\inf_{\sigma_{\ell}\in \curvespace{\ell}}\dtwp(\sigma_{\ell},\sigma)\leq \dtwp(\sigma^*,\sigma)\leq (1+\eps)\inf_{\sigma_{\ell}\in \curvespace{\ell}}\dtwp(\sigma_{\ell},\sigma).$$
\end{definition}

A \emph{range space} is defined as a pair of sets $(X,\RRR)$, where $X$ is the \textit{ground set} and $\RRR$ is the \textit{range set} which is a subset of the power set $\mathcal{P}(X)=\{X'|X'\subset X\}$.
Let $(X,\RRR)$ be a range space. For $Y\subseteq X$, we denote: $\RRR_{|Y}= \{ R \cap Y \mid R \in \RRR \}$. If $\RRR_{|Y} =\power(Y) $, then $Y$ is \textit{shattered} by $\RRR$. A key property of range spaces is the so called Vapnik-Chernovenkis dimension~\cite{Sau72, She72, VC71} (VC dimension) which for a range space $(X,\RRR)$ is the maximum cardinality of a shattered subset of $X$. 

We are interested in range spaces defined by balls by the $p$-DTW distance:
    We define the ($p$-)DTW ball, of given complexity $m \in \NN$ and radius $r \geq 0$, of a curve $\sigma\in \bX_{\ell}^d$ as $B_{r,m}^p(\sigma) =\{ \tau \in \bX_m^d \mid \dtwp(\sigma,\tau)\leq r\}$.
    Define the range set of $p$-DTW balls as 
    \(\RRR_{m,\ell}^p = \{B_{r,m}^p(\sigma)\mid\sigma\in\curvespace{\ell}, r>0\}\).
The \emph{$p$-DTW range space} is the range space $\mathcal{X}_{m,\ell} = \left(\bX_m^d, \RRR_{m,\ell}^p\right)$.

\begin{figure}
    \centering
    \includegraphics[width=\textwidth]{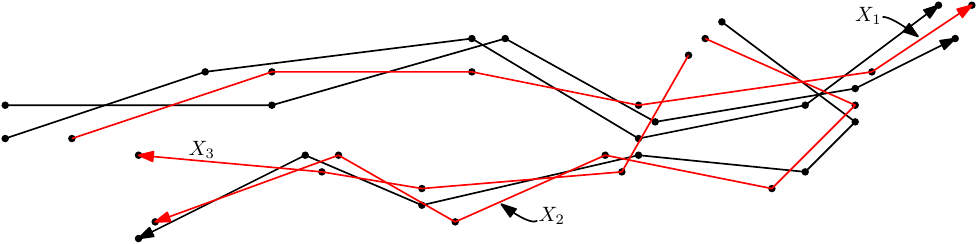}
    \caption{Illustration of a coreset (red), i.e. a weighted sparse representation of the original set of curves (in red and black). The weights in this case are $w(X_1)=3$, $w(X_2)=2$ and $w(X_3)=1$.}
    \label{fig:coreset}
\end{figure}

%

\section{VC Dimension of DTW}\label{sec:VCdim}

In this section, we derive bounds on the VC dimension of a range space that approximates the DTW range space. Our reasoning exclusively relies on establishing the prerequisites of \Cref{thm:vcpoly} below.
\begin{definition}[\cite{AB99}]
Let $H$ be a class of $\{0,1\}$-valued functions defined on a set $X$, and $F$ a class of real-valued functions defined on $\RR^d \times X$. We
say that H is a $k$-combination of $\sign(F)$ if there is a 
function
$g : \{-1,1\}^k \to  \{0,1\}$ and functions $f_1 , \ldots , f_k \in F $ so that for all $h \in H$
there is a parameter vector $\alpha \in \RR^d$ such that for all $x$ in $X$,
\[
h(x) = g(\sign(f_1(\alpha, x)),\ldots, \sign(f_k(\alpha, x))).
\]
\end{definition}
The definition for the sign function we use is that $\sign(x)=1$ for $\RR\owns x\ge 0$ and $\sign(x) = -1$ for $x<0$.
Observe that the class $H$ of functions corresponds to a system of subsets of $X$. 
\begin{theorem}[Theorem 8.3~\cite{AB99}]
\label{thm:vcpoly}
Let $F$ be a class of maps from $\RR^{s} \times X$ to $\RR$, so
that for all $x \in X$ and $f \in F$, the function $\alpha\mapsto  f( \alpha , x )$ is a polynomial
on $\RR^s$ of degree $\delta$. Let $H$ be a $\kappa$-combination of
$\sign(F)$. Then the VC dimension of $H$ is less than $2s \log_2 (12\delta\kappa)$.
\end{theorem}
\Cref{thm:vcpoly} implies a bound on the VC dimension of range spaces defined by $p$-DTW for even values of $p$, 
as follows (see \Cref{lem:VCDTWevenp} in \Cref{sec:vcproofs}). 
The decision of whether $p$-DTW exceeds a given threshold can be formulated as a $|\traversals_{m,\ell}|$-combination of signs of polynomial functions; each one realizing the cost of a traversal. This 
yields an upper bound of $O(d\ell^2\log (mp))$.
The situation becomes more intriguing in the general case, since for any odd $p$, the cost of each traversal is no longer a polynomial. To overcome this, we investigate range spaces defined by approximate $p$-DTW balls and show that we can get bounds that do not depend on $p$.

The following lemma shows that one can determine (approximately) the $p$-DTW between two sequences, based solely on the signs of certain polynomials, that are designed to provide a sketchy view of all point-wise distances.

\begin{restatable}{lemma}{dtwapprox}\label{lemma:new_dtwapprx}
    Let $\tau \in \curvespace{=\ell}$, $\sigma \in \curvespace{=m}$, $r > 0$ and $\epsilon\in (0,1]$. For each $i \in [\ell]$, $j\in [m]$ and $z \in [\lfloor \epsilon^{-1}+1\rfloor]$ define $$ f_{i,j,z}(\tau,r,\sigma) = \|\tau_i-\sigma_j\|^2 - (z\cdot \epsilon r)^2.$$ 
    There is an algorithm that, given as input the values of $\sign(f_{i,j,z}(\tau, r, \sigma))$, for all $i\in[\ell],j\in[m]$ and $z \in [\lfloor \epsilon^{-1}+1\rfloor]$, outputs a value in $\{0,1\}$ such that:
\begin{itemize}
    \item if $\dtwp(\tau,\sigma)\leq r$ then it outputs $1$, 
    \item if $\dtwp(\tau,\sigma) > (1+(m+\ell)^{1/p}\epsilon)r$ then it outputs $0$ and
    \item if $\dtwp(\tau,\sigma) \in (r, (1+(m+\ell)^{1/p}\epsilon)r]$ the output is either $0$ or $1$.
\end{itemize}
\end{restatable}

\begin{proof}
    The algorithm first iterates over all $i,j$. For each $i,j$, we assign a variable $\phi_{i,j}$ as follows: if $Z_{i,j}:=\{z \in [\lfloor \epsilon^{-1} +1 \rfloor] \mid \sign(f_{i,j,z}(\tau, r, \sigma)) = -1 \}\neq \emptyset$, then $\phi_{i,j} := \min(Z_{i,j})\epsilon r$, otherwise $\phi_{i,j} := \infty$. After having computed all $\phi_{i,j}$, we return a value as follows: 
    if \[\Phi(\tau,\sigma):=\min_{T \in \traversals _{\ell,m}} \left( \sum_{(i,j) \in T} (\phi_{i,j})^p \right)^{1/p} \leq (1+(m+\ell)^{1/p}\epsilon)r,\] then the output is $1$. Otherwise, the output is $0$. 
\\

We now prove the correctness of the algorithm. For this let us first observe that for all $i\in[\ell]$ and $j\in[m]$
it holds that $\|\tau_i - \sigma_j\| < \phi_{i,j}$. Further if $\|\tau_i - \sigma_j\|\leq r$ then $\phi_{i,j}-\epsilon r \leq \|\tau_i - \sigma_j\|$. This follows from the fact that $Z_{i,j}$ coincides with the set $\{z\in[\lfloor\epsilon^{-1}+1\rfloor]\mid z\epsilon r\geq ||\tau_i-\sigma_j||\}$.

For all $i\in[\ell]$ and $j\in[m]$, it holds that
$\|\tau_i - \sigma_j\| < \phi_{i,j}$, which implies that $\Phi(\tau,\sigma)\geq\dtw_p(\tau,\sigma)$, as for any $(\ell,m)$-traversal $T$ we see that
\[\left( \sum_{(i,j) \in T} (\phi_{i,j})^p \right)^{1/p} \geq \left( \sum_{(i,j) \in T} \|\tau_i - \sigma_j\|^p \right)^{1/p} \geq \dtw_p(\tau,\sigma).\]

It remains to show that if $\dtw_p(\tau,\sigma)\leq r$, then $\Phi(\sigma,\tau)\leq(1+(m+\ell)^{1/p}\epsilon)r$. 

For this, let $\dtw_p(\tau,\sigma)\leq r$ and let
$T^{\ast}$ be an $(\ell,m)$-traversal realizing $\dtw_p(\tau,\sigma)$. In particular, $\forall (i,j)\in T^{\ast}:~\|\tau_i - \sigma_j\|\leq r$, so that $\forall (i,j) \in T^{\ast}~\phi_{i,j} \leq \|\tau_i-\sigma_j\|+ \epsilon r$. We conclude that
\begin{align*}
 \Phi(\sigma,\tau) &\leq  \left( \sum_{(i,j) \in T^{\ast}} (\phi_{i,j})^p \right)^{1/p} \leq 
\left( \sum_{(i,j) \in T^{\ast}} \left|\|\tau_i-\sigma_j\|+\epsilon r\right|^p \right)^{1/p} \\
&\leq \left( \sum_{(i,j) \in T^{\ast}} \|\tau_i-\sigma_j\|^p\right)^{1/p}+ \left( \sum_{(i,j) \in T^{\ast}}\left(\epsilon r\right)^p \right)^{1/p}\\
&\leq r + \left( |T^{\ast}|\right)^{1/p} \cdot \epsilon r \leq r + \left(m+\ell\right)^{1/p}\epsilon r,
\end{align*}
where the inequalities hold by the Minkowski inequality and  $1\leq|T^{\ast}|\leq m+\ell$.
\end{proof}

The algorithm of \Cref{lemma:new_dtwapprx} essentially defines a function that implements approximate $p$-DTW balls membership, and satisfies the requirements set by \Cref{thm:vcpoly}. However, it is only defined on curves in $\curvespace{=\ell}$ and $\curvespace{=m}$. We extend the approach to all curves in $\curvespace{m}$.

\begin{restatable}{lemma}{approximaterange}\label{lemma:new_approximaterange}
    Let $\epsilon \in (0,1]$, and let $m,\ell\in\bN$ be given. There are injective functions $\pi_\ell:\curvespace{\ell}\rightarrow\bR^{(d+1)\ell}$ and $\pi_m:\curvespace{m}\rightarrow\bR^{(d+1)m}$ and a class of functions $F_{\epsilon}$ mapping from $\left(\bR^{(d+1)\ell}\times\bR\right)\times\bR^{(d+1)m}$ to $\bR$, such that for any $f\in F_{\epsilon}$, the function $\alpha \mapsto f(\alpha,x)$ is a polynomial function of degree $2$. Furthermore, there is a function $g \colon \{-1,1\}^{k} \to \{0,1\}$ and functions $f_1,\ldots,f_k \in F_{\epsilon}$, with $k = m\ell\lfloor\eps^{-1}+1\rfloor+m+\ell$, such that for any $\tau\in\curvespace{ \ell}$, $r>0$ and $\sigma\in \curvespace{m}$,  it holds that
    \begin{itemize}
        \item if $\dtwp(\sigma,\tau)\leq r$ then 
        \[g(\sign(f_1(\pi_\ell(\tau),r,\pi_m(\sigma))),\ldots,\sign(f_k(\pi_\ell(\tau),r,\pi_m(\sigma))))=1 ,
        \]
        \item if $\dtwp(\sigma,\tau) > (1+(m+\ell)^{1/p}\epsilon)r$ then
        \[
        g(\sign(f_1(\pi_\ell(\tau),r,\pi_m(\sigma))),\ldots,\sign(f_k(\pi_\ell(\tau),r,\pi_m(\sigma))))=0.
        \]
    \end{itemize}  
\end{restatable}
\begin{proof}
We first define $\widetilde{\sigma}= \pi_m(\sigma)$. $\widetilde{\sigma}$ consists of $m$ points in $\RR^{d+1}$. The first
$m'$ points consist of the points in $\sigma$ together with a $1$ in the $(d + 1)$th coordinate. The
$(m' + 1)$ to $m$th points consist of the point $(-1, \ldots , -1)$. The point $\widetilde{\tau}= \pi_{\ell} (\tau )$ is defined
similarly padding $\tau$ to a length of $\ell$ similar to $\sigma$.
Let $\tau_i$ and $\sigma_j$ denote the first $d$ coordinates of the $i$th and $j$th point in $\widetilde{\tau}$ and $\widetilde{\sigma}$ that
is for $j \leq m'$ the point $\sigma_j$ is exactly the $j$th point of $\sigma$, and for $j > m'$ the point  $\sigma_j$ is a
vector consisting of only $-1$’s. Further let $\tau_i^{d+1}$ denote the $(d + 1)$th coordinate of the $i$th
$(d + 1)$-dimensional point in $\widetilde{\tau}$, similarly for $\sigma_i^{d_1}$.

The set $F_{\epsilon}$ consists of all functions $f_{i,j,z} (\widetilde{\tau},r , \widetilde{\sigma}) = \|\tau_i -\sigma_j \|^2 - \left(\frac{z\epsilon}{(m+\ell)^{1/p}}r \right)^2$, where $i \in [\ell]$, 
$j \in [m]$ and $z\in [\lfloor \frac{(m+\ell)^{1/p}}{\epsilon}+1\rfloor]$. It further contains the functions $g_i (\widetilde{\tau},r , \widetilde{\sigma}) = \tau_i^{d+1}$ and
$h_j (\widetilde{\tau},r , \widetilde{\sigma}) = \sigma_j^{d+1}$.
The function $g$ has $k = \ell m \cdot \lfloor \epsilon^{-1} + 1\rfloor + m + l$ arguments, corresponding to the sign of
a function $f_{i,j,z}$, $g_i$ or $h_j$. Always ordered in the same way.
To compute $g$ we first use the sign of $g_i$ and $h_j$ to infer the values of $m'$ and $\ell'$, that is,
the complexities of $\sigma$ and $\tau$, as $\sign(g_i ) = 1$ if and only if $i < \ell'$, and similarly $\sign(h_j ) = 1$ if
and only if $j < m'$. It then invokes the algorithm of Lemma 9 with input
$\sign(f_{1,1,1} (\widetilde{\tau},r , \widetilde{\sigma})), \ldots , \sign(f_{\ell',m' ,\lfloor \epsilon^{-1} +1\rfloor}  (\widetilde{\tau},r , \widetilde{\sigma})).$
The statement then directly follows from \Cref{lemma:new_dtwapprx}.    
\end{proof}

We use the previous lemmas to define a distance function $\adtwp$ between elements of $\curvespace{m}$ and $\curvespace{\ell}$, which we will use throughout the paper as an approximate function of $\dtwp$. 
To get an estimate of the VC dimension of the range space induced by balls under $\adtwp$ and decide membership of points to these balls, the approximate distance will only take discrete values. 

\begin{definition}
\label{def:sampledtw}
Let $\eps\in(0,1]$ and define the set of radii $R_\eps=\{(1+\eps)^z\mid z\in\bZ\}$.  \Cref{lemma:new_approximaterange} defines an approximation of $\dtwp(\sigma,\tau)$ for any $\sigma\in\curvespace{\ell}$ and $\curvespace{m}$, by virtue of the functions $g$ and $f_1,...,f_k$ for $F_{\epsilon/(m+\ell)^{1/p}}$, as 
\[
\adtwp(\sigma,\tau)=(1+\eps)\cdot\sup\{r\in R_\eps\mid g(\sign(f_1(\pi_\ell(\tau),r,\pi_m(\sigma))),\ldots)=1\}.
\]
\end{definition}
    Overall, $\adtwp$ corresponds to the first value $r$ in $R_\eps$ for which the function $g$ of  \Cref{def:sampledtw} outputs a $0$. Note that the algorithm also outputs $0$ for all larger values in $R_\eps$. Notably, the function $g$ of \Cref{def:sampledtw} outputs $1$ for $r/(1+\eps)$. In the following lemma, we formally show that $\adtwp(\sigma,\tau)$ approximates $p$-DTW between $\sigma$ and $\tau$ within a factor of $1+\eps$. 
\begin{lemma}
\label{lem:sandwich}
    Let $0<\eps\leq 1$. For any $\sigma\in\curvespace{m}$ and $\tau\in\curvespace{\ell}$ it holds that \[\dtwp(\sigma,\tau)<\adtwp(\sigma,\tau)\leq(1+\eps)\dtwp(\sigma,\tau).\]
\end{lemma}
\begin{proof}
    Let $r=\adtwp(\sigma,\tau)\in R_\eps$. By definition the function $g$ of \Cref{def:sampledtw} outputs $1$ with $\sigma$, $\tau$ and $r/(1+\eps)$. Thus $\dtwp(\sigma,\tau)\leq (1+\eps)r/(1+\eps)=r$. As the algorithm outputs $0$ for $\sigma$, $\tau$ and $r$ it follows that $\dtwp(\sigma,\tau)> r/(1+\eps)$ implying the claim.
\end{proof}
Moreover, from the definition of $\adtwp$, we conclude that $g$ serves as a membership predicate for balls defined by $\adtwp$.  
\begin{lemma}

\label{lem:membership}
    Let $\eps\in(0,1]$, $\tau\in\curvespace{\ell}$ and $r\in R_\eps$. For any $\sigma\in\curvespace{m}$ the output of the function $g$ of \Cref{def:sampledtw} with $\sigma$, $\tau$ and $r$ corresponds to the decision whether the curve $\sigma$ is in the $r$-ball $\{x\in\curvespace{m}\mid\adtwp(x,\tau)\leq r\}$ centered at $\tau$.
\end{lemma}
\begin{proof}
    Let $r'=\adtwp(\sigma,\tau)\in R_\eps$. Assume $r'\leq r$ which by \Cref{lem:sandwich} implies that $\dtwp(\sigma,\tau)\leq r$. Then the function $g$ of \Cref{def:sampledtw} with $\sigma$, $\tau$ and $r$ outputs $1$. Now let $r<r'\in R_\eps$. However, in this case $g$ with $\sigma$, $\tau$ and $r$ will by definition of $\adtwp$ output $0$. Thus membership to a ball range corresponds to the output of the function $g$ of \Cref{def:sampledtw}.    
\end{proof}

We conclude with the main result of this section, namely an upper bound on the VC dimension of the range space that approximates the $p$-DTW range space. 

\begin{theorem}\label{thm:approxballsvcdim}
    Let $\eps\in (0,1]$ and $\widetilde{\RRR}_{m,\ell}^{p}=\{\{x\in\curvespace{m}\mid\adtwp(x,\tau)\leq r\}\subset\curvespace{m}\mid\tau\in\curvespace{\ell},r>0\}$ be the range set consisting of all balls centered at elements of $\curvespace{\ell}$ under $\adtwp$ in $\curvespace{m}$. The VC dimension of $(\curvespace{m},\widetilde{\RRR}_{m,\ell}^{p})$ is at most \[2(d+1)\ell\log_2(12\ell m\lfloor(m+\ell)^{1/p}\eps^{-1}+1\rfloor+12m+12\ell)=O(d\ell\log(\ell m\eps^{-1})).\]
\end{theorem}

\begin{proof}
    This follows from 
    \Cref{thm:vcpoly},  \Cref{lemma:new_approximaterange} and \Cref{lem:membership}, and the fact that any ball of radius $r>0$ under $\adtwp$ coincides with some ball with radius $\widetilde{r}\in R_\eps$ under $\adtwp$. Finally, the statement is implied by the injectivity of the functions $\pi_m$ and $\pi_{\ell}$.
    \end{proof}

In this section, we defined a distance function $\adtwp$ between curves in $\curvespace{m}$ and those in $\curvespace{\ell}$ that $(1+\eps)$-approximates $\dtwp$ and an upper bound on the VC dimension of the range space induced by balls of $\adtwp$, thereby producing an approximation of the $p$-DTW range space that we make use of below. The bound on the VC dimension is comparable to the one we get for exact $p$-DTW balls in certain cases (\Cref{lem:VCDTWevenp} in \Cref{sec:vcproofs}).
We emphasize that the sole purpose of $\adtwp$ is to obtain bounds on the size of a sample constituting a coreset through the knowledge of the VC dimension. At no point do we compute $\widetilde{\dtwp}$.

\section{Sensitivity bounds and coresets for DTW}\label{section:coresets}
To make use of the sensitivity sampling framework for coresets by Feldman and Langberg~\cite{Feldman2011}, we recast the input set $T\subset\curvespace{m}$ as a set of functions. Consider for any $y\in \curvespace{m}$ the real valued function $f_y$ defined on (finite) subsets of $\curvespace{\ell}$ by 
$f_y(C)=\min_{c\in C}\dtwp(y,c)$ for $C\subset \curvespace{\ell}$, transforming $T$ into $F_T=\{f_\tau\mid\tau\in T\}$. To construct a coreset, one draws elements from $T$ according to a fixed probability distribution over $T$, and reweighs each drawn element. Both the weight and sampling probability are expressed in terms of the \emph{sensitivity} of the drawn element $t$, which describes the maximum possible relative contribution of $t$ to the cost of any query evaluation. In our case, as we restrict a solution to a size of $k$, it turns out that it suffices to analyze the sensitivity with respect to inputs of size $k$. 
\begin{definition}[sensitivity]
    Let $F$ be a finite set of functions from $\power(\curvespace{\ell})\setminus\{\emptyset\}$ to $\bR$. For any $f\in F$ define the sensitivity
    \[\mathfrak{s}(f,F)=\sup_{C=\{c_1,\ldots,c_k\}\subset Z:\sum_{g\in F}g(C)>0}\frac{f(C)}{\sum_{g\in F}g(C)}.\]
    The total sensitivity $\mathfrak{S}(F)$ of $F$ is defined as $\sum_{f\in F}\mathfrak{s}(f,F)$.
\end{definition}
A crucial step in our approach is to show that any $(\alpha,\beta)$-approximation for $(k,\ell)$-median under $\dtwp$ can be used to obtain a bound on the total sensitivity associated to approximate distances.  
This is facilitated by the following lemma, that is a weaker version of the triangle inequality, as in general $\dtwp$ is not a metric (see \Cref{fig:path-angle}).
\begin{figure}
    \centering
    \includegraphics[width=\textwidth]{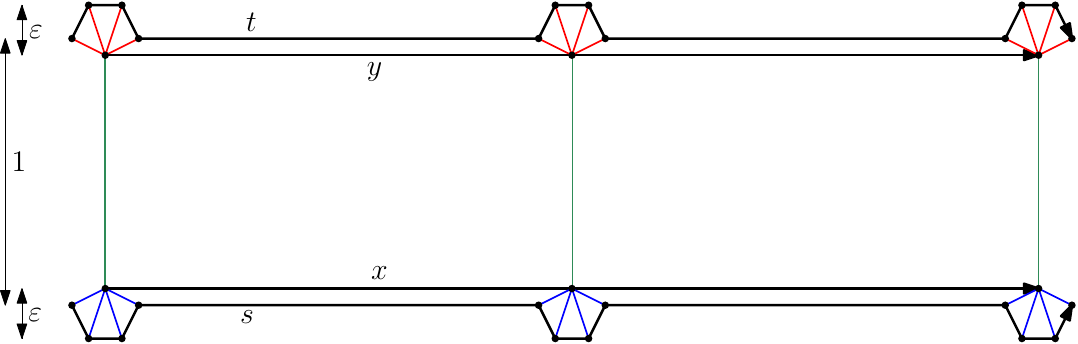}
    \caption{Violated triangle inequality as $\dtw(s,t)\approx12$, but $\dtw(s,x)\approx 0$ (matching in blue), $\dtw(y,t)\approx 0$ (red matching) and $\dtw(x,y)\approx 3$ (green matching).}
    \label{fig:path-angle}
\end{figure}
\begin{lemma}[weak triangle inequality \cite{L09}]
\label{lem:weaktrinagle}
    For two curves $x$ and $z$ of complexity $m>0$ and any curve $y$ of complexity $\ell>0$ it holds that 
    $\dtwp(x,z)\leq m^{1/p}(\dtwp(x,y)+\dtwp(y,z)).$
\end{lemma}
Note that the distance function in question is not $\dtwp$, but the $(1+\epsilon)$-approximation $\adtwp$ of DTW from before. 
For any $y\in \curvespace{m}$ and $\eps>0$, let $\widetilde{f}_y:\power(\curvespace{\ell})\setminus\{\emptyset\}\rightarrow\bR$ with $\widetilde{f}_y(C)=\min_{c\in C}\adtwp(y,c)$. 
Similarly, let $\widetilde{F}_T$ be the set $\{\widetilde{f}_\tau\mid\tau\in T\}$ for any $T\subset \curvespace{m}$.

\begin{restatable}{lemma}{epssensitivities}\label{lem:epssensitivities}
    Let $0<\eps\leq1$ and let $T\subset \curvespace{m}$ be the input of size $n$ for $(k,\ell)$-median and let $\hat{C}=\{\hat{c}_1,\ldots,\hat{c}_{\hat{k}}\}\subset \curvespace{\ell}$ be an $(\alpha,\beta)$-approximation to the $(k,\ell)$-median problem on $T$ with cost $\hat{\Delta}=\sum_{\tau\in T}\min_{\hat{c}\in\hat{C}}\dtwp(\tau,\hat{c})$, of size $\hat{k}\leq \beta k$. For any $i\in[\hat{k}]$ let $\hat{V}_i=\{\tau\in T\mid \dtwp(\tau,\hat{c}_i)=\min_{\hat{c}\in\hat{C}}\dtwp(\tau,\hat{c})\}$ be the Voronoi region of $\hat{c}_i$, the set of which (breaking ties arbitrarily) partitions $T$. Let $\hat{\Delta}_i=\sum_{\tau\in\hat{V}_i}\dtwp(\tau,\hat{c}_i)$ be the cost of $\hat{V}_i$. For all $\tau\in \hat{V}_i$ let
    \[\gamma(\widetilde{f}_\tau):=(m\ell)^{1/p}\left( \frac{2\alpha \dtwp(\tau,\hat{c}_i)}{\hat{\Delta}} + \frac{4}{\lvert\hat{V}_i\rvert} + \frac{8\alpha\hat{\Delta}_i}{\hat{\Delta}\lvert\hat{V}_i\rvert} \right).\]
    Then $\mathfrak{s}(\widetilde{f}_\tau,\widetilde{F}_T)\leq\gamma(\widetilde{f}_\tau)$ for any $\tau\in T$, and $\mathfrak{S}(\widetilde{F}_T)\leq\sum_{\tau\in T}\gamma(\widetilde{f}_\tau)\leq (m\ell)^{1/p}(4\hat{k} +10\alpha)$.
\end{restatable}

\begin{proof}
Fix an arbitrary set $C = \{ c_1, \dots, c_k\} \subseteq \curvespace{\ell}$, $i\in [\hat{k}]$ and $\tau\in \hat{V}_i$. 
Let further $\hat{B}_i = \{ \sigma \in \hat{V}_i \mid \dtwp(\sigma,\hat{c}_i) \leq 2\hat{\Delta}_i/\lvert \hat{V}_i \rvert\}$. Observe that for any $\tau\in\hat{B}_i$ it holds that $\widetilde{f}_\tau(\hat{C})\leq(1+\eps)2\hat{\Delta}_i/\lvert \hat{V}_i \rvert\leq4\hat{\Delta}_i/\lvert \hat{V}_i \rvert$. Breaking ties arbitrarily, let $c(x)\in C$ be the nearest neighbour of $x\in X$ among $C$.

We observe that $\lvert \hat{B}_i \rvert \geq \lvert \hat{V}_i \rvert/2$, since otherwise $\sum_{\sigma \in \hat{V}_i \setminus \hat{B}_i} \dtwp(\tau,\hat{c}_i) > \hat{\Delta}_i$. Additionally note that $\sum_{\widetilde{f}_\sigma \in \widetilde{F}_T} \widetilde{f}_\sigma(C) \geq \hat{\Delta}/\alpha$. 

For any $\sigma\in\hat{B}_i$ \Cref{lem:weaktrinagle} implies that $\widetilde{f}_\sigma(C) \geq\dtwp(\sigma,c(\sigma))\geq \frac{\dtwp(\hat{c}_i, c(\sigma))}{\ell^{1/p}} - \frac{4\hat{\Delta}_i}{\lvert\hat{V}_i\rvert}$. Now as $\dtwp(\hat{c}_i,c(\sigma))\geq\dtwp(\hat{c}_i,c(\hat{c}_i))$ it holds that
    \[\sum_{\widetilde{f}_\sigma \in \widetilde{F}_T} \widetilde{f}_\sigma(C) \geq \max\left\{\sum_{\sigma \in \hat{B}_i} 
    \widetilde{f}_\sigma(C) ,\frac{\hat{\Delta}}{\alpha}\right\}\geq \max\left\{\frac{\lvert\hat{V}_i\rvert}{2}\!\!\left(\frac{\dtwp(\hat{c}_i, c(\hat{c}_i))}{\ell^{1/p}} -\frac{4\hat{\Delta}_i}{\lvert\hat{V}_i\rvert}\right),\frac{\hat{\Delta}}{\alpha}\right\}.
    \]
    Assume that $\frac{\lvert\hat{V}_i\rvert}{2}\!\!\left(\frac{\dtwp(\hat{c}_i, c(\hat{c}_i))}{\ell^{1/p}} -\frac{4\hat{\Delta}_i}{\lvert\hat{V}_i\rvert}\right)\geq\frac{\hat{\Delta}}{\alpha}$, i.e. $\dtwp(\hat{c}_i, c(\hat{c}_i))\geq(2\ell^{1/p})\frac{\hat{\Delta} + 2\alpha\hat{\Delta}_i}{\alpha\lvert \hat{V}_i \rvert}=:\delta_i$.
    To eliminate the dependence on elements in $C$, we consider the function
    \[h_{i,\tau}:\left[\delta_i,\infty\right)\rightarrow \bR_{>0},x\mapsto \frac{2m^{1/p} (\ell^{1/p}\dtwp(\tau,\hat{c}_i) + x)}{\frac{\lvert \hat{V}_i \rvert x}{2\ell^{1/p}} - 2\hat{\Delta}_i}\]
    and observe that it is monotone and thus its maximum is either $h_{i,\tau}(\delta_i)$ or $\lim_{x\rightarrow\infty}h_{i,\tau}(x)$. Importantly, as \Cref{lem:weaktrinagle} implies that $\widetilde{f}_\tau(C)\leq (1+\eps)m^{1/p}(\dtwp(\tau,\hat{c}_i) + \dtwp(\hat{c}_i,c(\hat{c}_i)))$, we have $\frac{\widetilde{f}_\tau(C)}{\sum_{\widetilde{f}_\sigma\in \widetilde{F}_T}\widetilde{f}_\sigma(C)}\leq h_{i,\tau}(\dtwp(\hat{c}_i,c(\hat{c}_i)))$. If instead $\dtwp(\hat{c}_i, c(\hat{c}_i))<\delta_i$, then
    \[\frac{\widetilde{f}_\tau(C)}{\sum_{\widetilde{f}_\sigma\in \widetilde{F}_T}\widetilde{f}_\sigma(C)}\leq\frac{\widetilde{f}_\tau(C)}{\hat{\Delta}/\alpha}<\frac{(1+\eps)m^{1/p} (\dtwp(\tau,\hat{c}_i) + \delta_i)}{\hat{\Delta}/\alpha}\leq h_{i,\tau}(\delta_i).\]
    Thus it folows that
    \begin{align*}
        \mathfrak{s}(\widetilde{f}_\tau, \widetilde{F}_T) &= \sup_{C=\{c_1, \dots, c_k\} \subseteq Z}\frac{\widetilde{f}_\tau(C)}{\sum_{\widetilde{f}_\sigma\in \widetilde{F}_T}\widetilde{f}_\sigma(C)}
        \leq \max\left\{h_{i,\tau}(\delta_i),\lim_{x\rightarrow\infty}h_{i,\tau}(x)\right\}\\
        &=\max\left\{(m\ell)^{1/p}\left( \frac{2\alpha \dtwp(\tau,\hat{c}_i)}{\hat{\Delta}} + \frac{4}{\lvert\hat{V}_i\rvert} + \frac{8\alpha\hat{\Delta}_i}{\hat{\Delta}\lvert\hat{V}_i\rvert} \right),\frac{4(m\ell)^{1/p}}{\lvert\hat{V}_i\rvert}\right\} = \gamma(\widetilde{f}_\tau).
    \end{align*}
    Overall it follows that
    \[\mathfrak{S}(\widetilde{F}_T) \leq(m\ell)^{1/p}\sum_{i=1}^{k^\prime} \sum_{\sigma \in \hat{V}_i} \left( \frac{2\alpha \dtwp(\tau,\hat{c}_i)}{\hat{\Delta}} + \frac{4}{\lvert\hat{V}_i\rvert} + \frac{8\alpha\hat{\Delta}_i}{\hat{\Delta}\lvert\hat{V}_i\rvert} \right) = (m\ell)^{1/p}\left(2\alpha + 4\hat{k} + 8\alpha\right).\qedhere\]
\end{proof}

\begin{lemma}\label{lem:approxdistancecoreset}
A weighted $\epsilon$-coreset for $(k,\ell)$-median of $T$ under the approximate distance $\widetilde{\dtwp}$ is a weighted $3\epsilon$-coreset for $(k,\ell)$-median under $\dtwp$.
\end{lemma}
For $T = \{ \tau_1, \dots, \tau_n \} \subseteq \curvespace{m}$, we write $\widetilde{f}_i=\widetilde{f}_{\tau_j}$.
\begin{proof}
Denote by $\cost{T}{C}$ be the cost associated to $\dtwp$, $C\subset \curvespace{\ell}$ with $|C|=k$, and similarly denote the cost of $\widetilde{\dtwp}$ by $\acost{T}{C}$. We will show that 
\[
(1-3\epsilon)\cost{T}{C}\le \sum_{\tau_i\in S} w(\widetilde{f}_i)\min_{c\in C}\dtwp(\tau_i,c)= \sum_{\tau_i\in S} w(\widetilde{f}_i)f_i(C)\le (1+3\epsilon)\cost{T}{C}.
\]
Denote by $\scost{T}{C}=\sum_i w(\widetilde{f}_i)\min_{c\in C}\widetilde{\dtwp}(\tau_i,c)$ the cost associated to the weighted $\epsilon$-coreset for $\widetilde{\dtwp}$. Observe the relations
\begin{align*}
&(1-\epsilon)\acost{T}{C}\le \scost{S}{C} \le (1+\epsilon)\acost{T}{C},&&\\
&    \cost{T}{C} \leq \acost{T}{C} \leq (1+\epsilon) \cost{T}{C},&&\\
&\sum_{\tau_i\in S} w(\widetilde{f}_i)f_i(C)\le \sum_{\tau_i\in S} w(\widetilde{f}_i)\widetilde{f}_i(C)=\scost{T}{C}  \leq (1+\epsilon) \sum_{\tau_i\in S} w(\widetilde{f}_i)f_i(C),&&
\end{align*}
which together imply 
\begin{align*}
\cost{T}{C}\le\frac{1+\epsilon}{1-\epsilon}\cost{T}{C}\le \frac{1+\epsilon}{1-\epsilon} \acost{T}{C}\le \frac{\scost{S}{C}}{1-\epsilon} \le \sum_{\tau_i\in S} w(\widetilde{f}_i)f_i(C)&\\
\le \scost{S}{C}\le (1+\epsilon)\acost{T}{C}\le (1+\epsilon)^2\cost{T}{C}\le (1+3\epsilon)\cost{T}{C}.&
\end{align*}
\end{proof}

\begin{definition}[{\cite[Definition 2.3]{Har-Peled2011}}]
    \label{def:etaepsilonapprox}
    Let $\epsilon, \eta \in (0,1)$ and $(X,\mathcal{R})$ be a range space with finite non-empty ground set. An $(\eta, \epsilon)$-approximation of $(X,\mathcal{R})$ is a set $S \subseteq X$, such that for all $R \in \mathcal{R}$
    \begin{align*}
        \left\lvert \frac{\lvert R \cap X \rvert}{\lvert X \rvert} - \frac{\lvert R \cap S \rvert}{\lvert S \rvert} \right\rvert \leq 
        \begin{cases}
            \epsilon \cdot \frac{\lvert R \cap X \rvert}{\lvert X \rvert}, & \text{if } \lvert R \cap X \rvert \geq \eta \cdot \lvert X \rvert \\
            \epsilon \cdot \eta, & \text{else}.
        \end{cases}
    \end{align*}
\end{definition}
We employ the following theorem for obtaining $(\epsilon, \eta)$-approximations.
\begin{theorem}[{\cite[Theorem 2.11]{Har-Peled2011}}]
    \label{theo:etaepsilonapprox}
    Let $(X,\mathcal{R})$ be a range space with finite non-empty ground set and VC dimension $\mathcal{D}$. Also, let $\epsilon, \delta, \eta \in (0,1)$. There is an absolute constant $c \in \mathbb{R}_{>0}$ such that a sample of $$ \frac{c}{\eta \cdot \epsilon^2} \cdot \left( \mathcal{D} \log \left(\frac{1}{\eta}\right) + \log\left(\frac{1}{\delta}\right) \right) $$ elements drawn independently and uniformly at random with replacement from $X$ is a $(\eta, \epsilon)$-approximation for $(X, \mathcal{R})$ with probability at least $1 - \delta$.
\end{theorem}

\begin{restatable}{theorem}{coresetconstruction}\label{theo:coreset}
	
	For \(\widetilde{f} \in \widetilde{F}\), let \(\lambda(\widetilde{f}) = 2^{\lceil \log_2(\gamma(\widetilde{f})) \rceil} \), with $\gamma(\widetilde{f})$ the sensitivity bound of \Cref{lem:epssensitivities}, associated to an $(\alpha,\beta)$-approximation consisting of $\hat{k}\leq \beta k$ curves, for $(k,\ell)$-median for curves in $\curvespace{m}$ under $\dtwp$, \(\Lambda = \sum_{\widetilde{f} \in\widetilde{F}} \lambda(\widetilde{f})\), \(\psi(\widetilde{f}) = \frac{\lambda(\widetilde{f})}{\Lambda}\) 
    and \(\delta, \epsilon \in (0,1)\). A sample \(S\) of 
 \[
 \Theta\hspace{-0.2em}\left(\epsilon^{-2} \alpha \hat{k} (m\ell)^{1/p}\left((d\ell\log(\ell m\epsilon^{-1})) k \log(k) \log(\alpha n) \log(\alpha \hat{k} (m\ell)^{1/p}) + \log(1/\delta)\right)\right)
 \]
 elements \(\tau_i \in T\), drawn independently with replacement with probability \(\psi(\widetilde{f}_i)\) and weighted by \(w(\widetilde{f}_i) = \frac{\Lambda}{\lvert S \rvert \lambda(\widetilde{f}_i)}\) is a weighted \(\epsilon\)-coreset for \((k,\ell)\)-median clustering of $T$ under $\dtwp$ with probability at least \(1 - \delta\). 
\end{restatable}
Let $\acost{T}{C}=\sum_{\tau \in T}\widetilde{f}_\tau(C)$ for $T = \{ \tau_1, \dots, \tau_n \} \subseteq \curvespace{m}$, and denote $\widetilde{F} = \{ \widetilde{f}_1, \dots, \widetilde{f}_n \}$, with $\widetilde{f}_i=\widetilde{f}_{\tau_j}$.

 Our proof relies on the reduction to uniform sampling, introduced by~\cite{Feldman2011} and improved by~\cite{DBLP:journals/corr/BravermanFL16}, allowing us to apply \cref{theo:etaepsilonapprox}. In the following, we adapt and modify the proof of Theorem 31 in~\cite{DBLP:journals/siamcomp/FeldmanSS20} and combine it with results from~\cite{DBLP:conf/nips/MunteanuSSW18} to handle the involved scaling, similarly to Theorem 4 in~\cite{DBLP:conf/caldam/BuchinR22}.
\begin{proof}
The proof of the theorem depends on analyzing different estimators for $\acost{T}{C}$ for \(C \subseteq \curvespace{\ell}\) arbitrary with \(\lvert C \rvert = k\). Consider first the estimator 
\[
\scost{S}{C} = \sum_{\tau_i \in S} w(\widetilde{f}_i) \cdot \min_{c \in C} \widetilde{\dtwp}(\tau_i, c) = \sum_{\tau_i \in S} w(\widetilde{f}_i) \cdot \widetilde{f}_i(C) = \sum_{\tau_i \in S} \frac{\Lambda}{\lvert S \rvert\lambda(\widetilde{f}_i)} \widetilde{f}_i(C)
\] 
for \(\acost{T}{C}\). 
We see that \(\scost{S}{C}\) is unbiased by virtue of 
\[\expected{\scost{S}{C}} = \sum_{i=1}^{\lvert S \rvert} \sum_{\tau_j \in T} \psi(\widetilde{f}_j)\frac{\Lambda}{\lvert S \rvert\lambda(\widetilde{f}_j)} \widetilde{f}_j(C) = \sum_{i=1}^{\lvert S\rvert}\sum_{\tau_j \in T} \frac{\widetilde{f}_j(C)}{\lvert S\rvert} = \acost{T}{C}.\] 	
We next reduce the sensitivity sampling to uniform sampling by letting \(G\) be a multiset that is a copy of \(\widetilde{F}\), where each \(\widetilde{f} \in \widetilde{F}\) is contained \(\lvert \widetilde{F} \rvert \lambda(\widetilde{f})\) times and is scaled by \(\frac{1}{\lvert \widetilde{F} \rvert \lambda(\widetilde{f})}\), so that \(\lvert G \rvert = \lvert \widetilde{F} \rvert \Lambda\) and \(\psi(\widetilde{f}) = \frac{\lvert \widetilde{F} \rvert \lambda(\widetilde{f})}{\lvert G \rvert}\). We clearly have 
\[\sum_{g \in G} g(C) = \sum_{\widetilde{f} \in \widetilde{F}} \frac{\lvert \widetilde{F} \rvert \lambda(\widetilde{f})}{\lvert \widetilde{F} \rvert \lambda(\widetilde{f})} \widetilde{f}(C) = \sum_{\widetilde{f} \in \widetilde{F}} \widetilde{f}(C) = \acost{T}{C}.\]
	
For a sample \(S^\prime\), with \(\lvert S^\prime \rvert = \lvert S \rvert\), drawn independently and uniformly at random with replacement from \(G\), consider the estimator for \(\acost{T}{C}\) defined by \[\ucost{S^\prime}{C} = \frac{\lvert G \rvert}{\lvert S^\prime \rvert} \sum_{g \in S^\prime} g(C),\]  where again \(C \subseteq \curvespace{\ell}\) with \(\lvert C \rvert = k\). We see that $\ucost{S^\prime}{C}$ is unbiased by virtue of
	\begin{align*}
		\expected{\ucost{S^\prime}{C}} & = \frac{\lvert G \rvert}{\lvert S^\prime \rvert} \sum_{i=1}^{\lvert S^\prime \rvert} \sum_{\widetilde{f} \in \widetilde{F}} \frac{\widetilde{f}(C)}{\lvert \widetilde{F} \rvert \lambda(\widetilde{f})} \frac{\lvert \widetilde{F} \rvert \lambda(\widetilde{f})}{\lvert G \rvert} = \frac{1}{\lvert S^\prime \rvert} \sum_{i=1}^{\lvert S^\prime \rvert} \sum_{\widetilde{f} \in \widetilde{F}} \widetilde{f}(C) =  \frac{\lvert S^{\prime}\rvert}{\lvert S^{\prime}\rvert}\acost{T}{C}.
	\end{align*}
We now assume that \(S^\prime = \left\{ \frac{1}{\lvert \widetilde{F} \rvert\lambda(\widetilde{f}_i)} \cdot \widetilde{f}_i \middle| \tau_i \in S \right\}\), which yields 
\begin{align*}
\ucost{S^\prime}{C} = \frac{\lvert \widetilde{F} \rvert \Lambda}{\lvert S^\prime \rvert}\sum_{g \in S^\prime} g(C) = \sum_{\tau_i \in S} \frac{ \Lambda}{\lvert S \rvert \lambda(\widetilde{f}_i)} \widetilde{f}_i(C)  = \scost{S}{C}. \tag{I}\label{eq:equalcosts}
\end{align*}
	
For \(H \subseteq G\), \(C \subseteq \curvespace{\ell}\) with $|C|=k$ and \(r \in \mathbb{R}_{\geq 0}\), we define \(\range{H}{C}{r} = \range{G}{C}{r}\cap H=\{ g \in H \mid g(C) \geq r \}\).
For all such \(C \subseteq Z\) and all \(H \subseteq G\), we have that
	\begin{align*}
		\sum_{g \in H} g(C) & = \sum_{g \in H} \int_0^\infty \ind{g(C) \geq r}\ \mathrm{d}r  
		 =\int_0^\infty \lvert \range{H}{C}{r} \rvert \ \mathrm{d}r, \tag{II}\label{eq:intidentity}
	\end{align*}
where all of the involved functions are integrable.	Consider now the range space \((G, \mathcal{R})\) over \(G\), where \(\mathcal{R} = \{ \range{G}{C}{r} \mid r \in \mathbb{R}_{\geq 0}, C \subseteq Z, \lvert C \rvert = k \}\). For the following, we apply \cref{theo:etaepsilonapprox} with the given \(\delta\), \(\epsilon/2\) and \(\eta = 1/\Lambda\), so as to guarantee that $S^{\prime}$ is a $(1/\Lambda,\epsilon/2)$-approximation of $(G,\mathcal{R})$. 
Given \(C \subseteq Z\) with \(\lvert C \rvert = k\), we compute that 
	\begin{align*}
		\left\lvert \acost{T}{C} - \scost{S}{C} \right\rvert & \stackrel{\eqref{eq:equalcosts}}{=} \left\lvert \acost{T}{C} - \ucost{S^\prime}{C} \right\rvert = \left\lvert \sum_{g \in G} g(C) - \frac{\lvert G \rvert}{\lvert S^\prime \rvert} \sum_{g \in S^\prime} g(C) \right\rvert \\ 
		& = \left\lvert \int_0^\infty \lvert \range{G}{C}{r} \rvert \ \mathrm{d} r -  \frac{\lvert G \rvert}{\lvert S^\prime \rvert} \int_0^\infty \lvert \range{S^\prime}{C}{r} \rvert \ \mathrm{d} r \right\rvert \\
		& = \left\lvert \int_0^\infty \lvert \range{G}{C}{r} \rvert - \frac{\lvert G \rvert}{\lvert S^\prime \rvert} \lvert \range{S^\prime}{C}{r} \rvert \ \mathrm{d}r \right \rvert \\
		& \leq \int_0^\infty \left\lvert \lvert \range{G}{C}{r} \rvert - \frac{\lvert G \rvert}{\lvert S^\prime \rvert} \lvert \range{S^\prime}{C}{r} \rvert \right \rvert \mathrm{d}r.
	\end{align*}
	
 	The monotonicity of \(\lvert \range{G}{C}{r} \rvert\) implies that \(R_1(C) = \{ r \in \mathbb{R}_{\geq 0} \mid \lvert \range{G}{C}{r} \rvert \geq \eta \cdot \lvert G \rvert \}\) and \(R_2(C) = \mathbb{R}_{\geq 0} \setminus R_1(C)\) are intervals. Denoting \(r_u(C) = \max\limits_{g \in G} g(C)\), we have that for \(r \in (r_u(C), \infty)\), it holds that \(\lvert \range{G}{C}{r} \rvert = 0\). Therefore, 
 	\[
 	\left\lvert \frac{\lvert\range{G}{C}{r} \rvert}{\lvert G \rvert} - \frac{\lvert \range{S^\prime}{C}{r} \rvert }{\lvert S^\prime \rvert} \right \rvert\le\frac{\epsilon}{2}\frac{\lvert\range{G}{C}{r} \rvert}{\lvert G \rvert},
 		\]
for $r\in R_1(C)$, since $S^{\prime}$ is a $(1/\Lambda,\epsilon/2)$-approximation for $(G,\mathcal{R})$ and similarly for $r\in R_2(C)$. Thus, 
	\begin{align*}
\left\lvert \acost{T}{C} - \scost{S}{C} \right\rvert & \le \int_{R_1}  \left\lvert \lvert \range{G}{C}{r} \rvert - \frac{\lvert G \rvert}{\lvert S^\prime \rvert} \lvert \range{S^\prime}{C}{r} \rvert \right \rvert \mathrm{d}r + \int_{R_2} \frac{\epsilon}{2}\eta \lvert G\rvert  \mathrm{d}r \\
		& \leq \frac{\epsilon}{2} \int\limits_{0}^{\infty} \lvert \range{G}{C}{r} \rvert \ \mathrm{d}r + \frac{\epsilon \eta \lvert G \rvert}{2} \int\limits_{0}^{r_u(C)} \ \mathrm{d} r \\
		& = \frac{\epsilon}{2} \sum_{g \in G} g(C) + \frac{\epsilon \eta \lvert G \rvert r_u(C)}{2}, \tag{III} \label{eq:half_bound}
	\end{align*}
	where the last equality is due to~\eqref{eq:intidentity}. We now bound the last term in~\eqref{eq:half_bound} with the help of the sensitivity bounds derived in \Cref{lem:epssensitivities}, where we use that $\gamma(\widetilde{f})\le \lambda(\widetilde{f})$. For each \(g \in G\), we have
	\begin{align*}
		\frac{g(C)}{\sum_{g \in G} g(C)} = \frac{\frac{1}{\lvert \widetilde{F} \rvert \lambda(\widetilde{f})} \widetilde{f}(C)}{\sum_{\widetilde{f} \in \widetilde{F}} \widetilde{f}(C)} \leq \frac{1}{\lvert \widetilde{F} \rvert \lambda(\widetilde{f})}\lambda(\widetilde{f}) = \frac{1}{\lvert \widetilde{F} \rvert},
	\end{align*}
	where \(\widetilde{f} \in \widetilde{F}\) is the function that \(g\) is a copy of, implying that \(r_u(C) \leq \frac{1}{\lvert \widetilde{F} \rvert} \sum_{g \in G} g(C)\). Thus, 
	\begin{align*}
		\frac{\epsilon \eta \lvert G \rvert r_u(C)}{2} & \leq \frac{\epsilon}{2} \frac{1}{\Lambda} \lvert \widetilde{F} \rvert \Lambda \frac{1}{\lvert \widetilde{F} \rvert} \sum_{g \in G} g(C) = \frac{\epsilon}{2} \sum_{g \in G} g(C).
	\end{align*}
	All in all,~\eqref{eq:half_bound} implies that \(\lvert \acost{T}{C} - \scost{S}{C} \rvert \leq \epsilon \cdot \sum_{g \in G} g(C)=\epsilon \cdot \acost{T}{C}\) for all \(C \subseteq Z\) with \(\lvert C \rvert = k\), with probability at least $1-\delta$, so $S$ is an $\epsilon$-coreset for the approximate distance function. 
	\\
	By \Cref{lem:approxdistancecoreset}, upon rescaling $\epsilon$ by $1/3$, it remains to show the asserted bounds on the size of $S^{\prime}$. By the use of \Cref{theo:etaepsilonapprox}, we need a bound on both $\Lambda$ and the VC dimension of the range space $(G,\mathcal{R})$. We observe that 
 \(\gamma(\widetilde{f}) \leq \lambda(\widetilde{f}) \leq  2\cdot \gamma(\widetilde{f})\), so that $\Lambda\le 2\Gamma(\widetilde{F})=O(\alpha \hat{k}(m\ell)^{1/p})$, by \Cref{lem:epssensitivities}.
	Moreover, the VC dimension of $(G,\mathcal{R})$ lies in $O(d\ell\log(\ell m \epsilon^{-1})k\log(k)\log(\alpha n))$ by \Cref{lem:vcdimensionrangespace} below, concluding the proof. 
\end{proof}

\begin{lemma}\label{lem:vcdimensionrangespace}
The VC dimension of the range space $(G,\mathcal{R})$ lies in $O(\mathcal{D}k\log(k)\log(\alpha n))$, with $\mathcal{D}=O(d\ell \log(\ell m \epsilon^{-1}))$ the VC dimension of $(\curvespace{m},\widetilde{\mathcal{R}}^p_{m\ell})$.
\end{lemma}
The proof is an adaptation of the methods of \cite[Lemma~11]{DBLP:conf/nips/MunteanuSSW18} and \cite[Lemma~2]{DBLP:conf/caldam/BuchinR22}.
\begin{proof}
We first consider the simple case that for all $\widetilde{f}\in \widetilde{F}$, $\lambda(\widetilde{f})=\widetilde{c}$, so that the scaling of the elements of $G$ is uniform and can be ignored in the context of the VC dimension. For $r\ge 0$, $c\in\curvespace{\ell}$, let $\widetilde{B}_{r,X}(c)=\{ \sigma\in\curvespace{m}|\widetilde{\dtwp}(\sigma,c)\le r\}\cap X$. The range space $(G,\mathcal{R})$ can then alternatively be described as $(T,\{T\setminus\bigcup_{c\in C}\widetilde{B}_{r,T}(c))|C\subset \curvespace{\ell}, |C|=k, r\in \mathbb{R}_{\ge 0}\})$, which in turn has VC dimension at most equal to that of $(\curvespace{m},\{\curvespace{m}\setminus\bigcup_{i\in [k]}\widetilde{B}_i|\widetilde{B}_1,...,\widetilde{B}_k\})$, with each $\widetilde{B}_i$ of the form $\widetilde{B}_{r,\curvespace{m}}(c)$ for a $c\in \curvespace{\ell}$. The last range space has the same VC as its complementary range space, which by the $k$-fold union theorem has VC dimension at most $ 2\mathcal{D}k\log_2 (3k)\le c\mathcal{D}k\log k \in O(\mathcal{D}k\log k)$~\cite[Lemma~3.2.3]{Blumer1989}. 

Let now $t$ denote the number of distinct values $\{c_1,...,c_t\}$ of $\lambda(\widetilde{f})$, as $\widetilde{f}$ ranges over $\widetilde{F}$ and partition $G$ into the sets $\{G_1,...,G_t\}$ such that for all $g\in G_i$ there is a $\widetilde{f}\in \widetilde{F}$ with $g=\frac{1}{|\widetilde{F}|\lambda(\widetilde{f})}\widetilde{f}=\frac{1}{|\widetilde{F}|c_i}\widetilde{f}$. Assume, for the sake of contradiction, that $G'\subset G$ is a set with $|G'|>t\cdot c \mathcal{D}k\log k$ that is shattered by $\mathcal{R}$. Consider the sets $G'_i=G'\cap G_i$ as well as induced range spaces $\mathcal{R}_i=G_i\cap \mathcal{R}$ on each $G_i$ for $i\in [t]$. Since the $G_i$ are disjoint, each $G'_i$ is shattered by $\mathcal{R}_i$ and there must exist at least one $j\in [t]$ such that $|G_j'|\ge\frac{|G'|}{t}>\frac{t\cdot c \mathcal{D}k\log k}{t}=c\mathcal{D}k\log k$. However, this contradicts the VC dimension of $(G_j,\mathcal{R}_j)$ in the case that the scaling of the functions in $G$ is uniform, established above. We now derive explicit bounds on $t$. By \Cref{lem:epssensitivities}, for $\tau\in \hat{V}_i$ with $i\in [k']$, we have
\[
(m\ell)^{1/p}\frac{4}{\lvert \hat{V}_i\rvert} \le\gamma(\widetilde{f})\le (m\ell)^{1/p}\left( 2\alpha + \frac{4}{\lvert\hat{V}_i\rvert} + \frac{8\alpha}{\lvert\hat{V}_i\rvert} \right),
\]
implying that the number of distinct values of $\lambda(\widetilde{f})$ is that number for $\lceil \log_2(\gamma(\widetilde{f}))\rceil$, which is 
\begin{align*}
&\log_2\left( (m\ell)^{1/p}\left( 2\alpha + \frac{4}{\lvert\hat{V}_i\rvert} + \frac{8\alpha}{\lvert\hat{V}_i\rvert} \right)\right) -\log_2\left((m\ell)^{1/p}\frac{4}{\lvert \hat{V}_i\rvert}\right)\\
= & \log_2\left( \left( 2\alpha + \frac{8\alpha}{\lvert\hat{V}_i\rvert} +\frac{4}{\lvert\hat{V}_i\rvert} \right) \left(\frac{4}{\lvert \hat{V}_i\rvert}\right)^{-1}\right)\le \log_2\left( \frac{n\alpha}{2}+2\alpha+1\right).
\end{align*} 
Thus, the VC dimension of $(G,\mathcal{R})$ is at most $2\mathcal{D}k\log_2 (3k)\log_2(\frac{n\alpha}{2} +2\alpha+1)$. 
\Cref{thm:approxballsvcdim} concludes the proof.
\end{proof}

We remark that in the limit $p\to \infty$, the constructed coreset has a very similar size as a recent construction for coresets for the Fréchet distance~\cite{DBLP:conf/caldam/BuchinR22}.

\section{Linear time $(O((m\ell)^{1/p}),1)$-approximation algorithm for $(k,\ell)$-median}
In this section, we develop approximation algorithms for $(k,\ell)$-median for a set $T\subset \curvespace{m}$ of $n$ curves. 
For this, we approximate DTW on $T$ by a metric using a new inequality for DTW (\Cref{lem:quadineq}). This allows the use of any approximation algorithm for $k$-median in metric spaces, leading to a first approximation algorithm of the original problem. 
However, computing the whole metric space would take $O(n^3)$ time. We circumvent this by in turn using the DTW distance to approximate the metric space. Combined with a $k$-median algorithm in metric spaces~\cite{DBLP:conf/stoc/Indyk99}, we obtain a linear time $(O((m\ell)^{1/p}),1)$-approximation algorithm.


\subsection{Dynamic time warping approximating metric}
We begin with the following more general triangle inequality for $\dtwp$, which motivates analysing the metric closure of the input set. While $\dtwp$ does not satisfy the triangle inequality (see \Cref{fig:path-angle}), the inequality shows it is never `too far off'. Remarkably, the inequality does not depend on the complexity of the curves `visited'. The missing proofs in this section are deferred to \Cref{sec:metricapproximation}. \Cref{lem:quadineq} is illustrated in \Cref{fig:inducedtraversal}.

\begin{restatable}[Iterated triangle inequality]{lemma}{iteratedtriinequality}
\label{lem:quadineq}
    Let $s\in \curvespace{\ell}$, $t\in \curvespace{\ell'}$ and $X=(x_1,\ldots,x_r)$ be an arbitrary ordered set of curves  in $\curvespace{m}$. Then
    \[\dtwp(s,t)\leq (\ell+\ell')^{1/p}\left(\dtwp(s,x_1)+\sum_{i<r}\dtwp(x_i,x_{i+1})+\dtwp(x_r,t)\right).\]
\end{restatable}
\begin{proof}
    To ease exposition assume, that $r=2$, that is $X=(x,y)$.
    Let $W_{sx}$ be an optimal traversal of $s$ and $x$ realising $\dtwp(s,x)$. Similarly define $W_{xy}$ and $W_{yt}$. From this we now construct a traversal of $s$ and $t$ endowed with additional information on which vertices of $x$ and $y$ where used to match the vertices of $s$ and $t$. More precisely we will construct an ordered set $W$ of indices $((\alpha_1,\beta_1,\gamma_1,\delta_1),(\alpha_2,\beta_2,\gamma_2,\delta_2),\ldots)$, such that for any $i\geq 2$ it holds that $(\alpha_{i},\delta_{i})-(\alpha_{i-1},\delta_{i-1})\in\{(0,1),(1,0),(1,1)\}$, and for any $i\geq 1$ it holds that $(\alpha_i,\beta_i)\in W_{sx}$, $(\beta_i,\gamma_i)\in W_{xy}$ and $(\gamma_i,\delta_i)\in W_{yt}$. Refer to \Cref{fig:inducedtraversal} for a schematic view of the constructed set $W$.\\
    We begin with $W=((1,1,1,1))$, which clearly has the stated properties as $(1,1)$ is in any traversal. Now recursively define the next element in $W$ based on the last element $(\alpha,\beta,\gamma,\delta)$ of $W$.\\
    If $\alpha=\ell$ and $\delta=\ell'$ we stop adding elements to $W$. Otherwise, if $\delta<\ell'$ let $\delta'=\delta+1$. From this let $\gamma'=\min\{j\geq\gamma\mid(j,\delta')\in W_{yt}\}$, which exists, because $W_{yt}$ itself is a traversal. Similarly from $\gamma'$ define $\beta'$ and from $\beta'$ define $\alpha'$. If $\alpha'\leq\alpha+1$, then $(\alpha',\beta',\gamma',\delta')$ is added to $W$, and the steps are recursively repeated. Observe that $\alpha'\geq\alpha$, which implies that the properties of $W$ are preserved. If instead $\alpha'>\alpha+1$, let $\alpha''=\alpha+1$. From this define $\beta''=\min\{b\geq\beta\mid(\alpha'',b)\in W_{sx}\}$. Clearly $\beta''<\beta'$, as $\alpha''<\alpha'$. Similarly from this define $\gamma''$ for which it holds that $\gamma''<\gamma'$ and from this define $\delta''$ for which it holds that $\delta''<\delta'=\delta+1$. But by definition $\delta\leq\delta''$, thus $\delta=\delta''$. Thus we add $(\alpha'',\beta'',\gamma'',\delta'')$ preserving the properties of $W$. From here recursively repeat the steps above.\\
    In the case where $\delta=\ell'$, we set $\alpha'=\alpha+1$. From this we similarly define $\beta'=\min\{b\geq\beta\mid(\alpha',b)\in W_{sx}\}$, from which we similarly define $\gamma'$, from which we define $\delta'$. But as $\ell'=\delta\leq\delta'\leq\ell'$ it follows that $\delta'=\ell'$ and thus adding $(\alpha',\beta',\gamma',\delta')$ to $W$ also preservers its properties. From here recursively repeat the steps above.\\
    Observe now that
    \[\left(\sum_{(\alpha,\beta,\gamma,\delta)\in W}\|s_\alpha-x_\beta\|_2^p\right)^{1/p}\leq \left(\sum_{(\alpha,\beta)\in W_{sx}}|W|\cdot\|s_\alpha-x_\beta\|_2^p\right)^{1/p}=|W|^{1/p}\dtw(s,x)_p.\]
    Similarly for $x$ and $y$, and $y$ and $t$. Further, we aquire a traversal $\widetilde{W}=((\alpha_1,\delta_1),\ldots)$ of $s$ and $t$ from $W$ by ignoring the middle two indices of each element of $W$. Now overall
    \begin{align*}
        \dtwp(s,t)&\leq\left(\sum_{(\alpha,\delta)\in\widetilde{W}}\|s_\alpha-t_\delta\|_2^p\right)^{1/p}=\left(\sum_{(\alpha,\beta,\gamma,\delta)\in W}\|s_\alpha-t_\delta\|_2^p\right)^{1/p}\\
        &\leq\left(\sum_{(\alpha,\beta,\gamma,\delta)\in W}(\|s_\alpha-x_\beta\|_2+\|x_\beta-y_\gamma\|_2+\|y_\gamma-t_\delta\|_2)^p\right)^{1/p}\\
        &\leq |W|^{1/p}\dtw(s,x)_p +|W|^{1/p}\dtw(x,y)_p + |W|^{1/p}\dtw(y,t)_p,
    \end{align*}
    which concludes the proof, as $|W|=|\widetilde{W}|\leq\ell+\ell'$.

    Observe that this extends to $r>2$ in a straightforward way. In this case $W$ consists of tuples of length $2+r$, while $|W|=|\widetilde{W}|\leq \ell+\ell'$ still holds.
\end{proof}

\begin{figure}
    \centering
    \includegraphics[width=\textwidth]{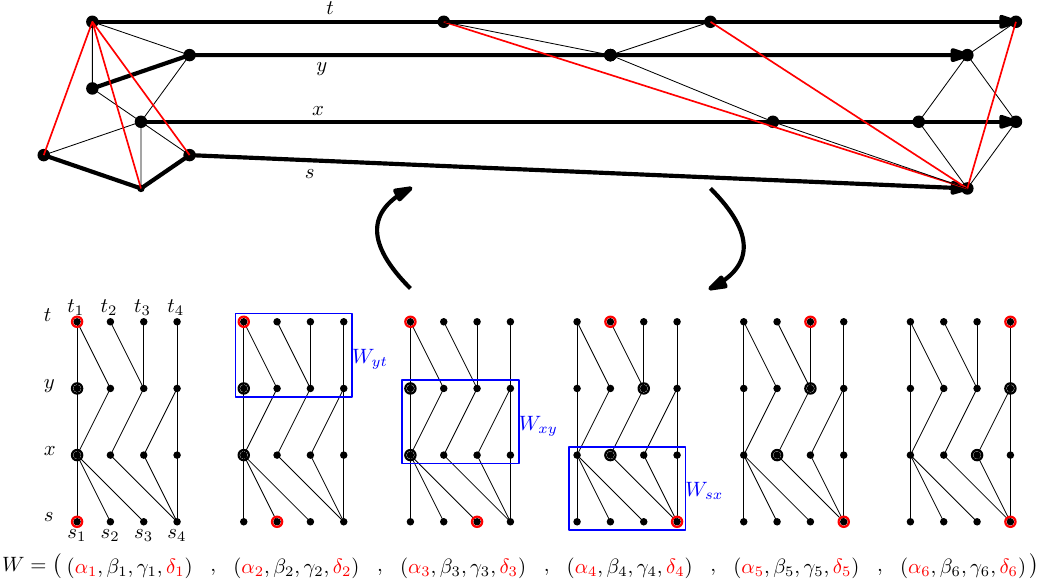}
    \caption{Illustration of how the optimal traversals $W_{sx}$, $W_{xy}$ and $W_{yt}$ of visited curves can be `composed' to yield a set $W$ that induces a traversal $\widetilde{W}$ (in red) of $s$ and $t$. Any single matched pair of vertices in $W_{sx}$, $W_{xy}$ or $W_{yt}$ is at most $|W|\leq\ell+\ell'$ times a part of $W$.}
    \label{fig:inducedtraversal}
\end{figure}
\begin{definition}[metric closure]
    Let $(X,\dist)$ be a finite set endowed with a distance function $\dist:X\times X\rightarrow\bR$. The metric closure $\overline{\dist}$ of $\dist$ is the function
    \[\overline{\dist}:X\times X\rightarrow \bR, (s,t)\mapsto\min_{\substack{r\geq 2,\{\tau_1,\ldots,\tau_r\}\subset X \\ s=\tau_1,t=\tau_r}}\sum_{i<r}\dist(\tau_i,\tau_{i+1}).\]
\end{definition}
The metric closure of any distance function is a semimetric and can be extended to a metric by removing duplicates or small (symbolic) perturbations. 
Note that the metric closure of $\dtwp$ can be strictly smaller than $\dtwp$ because DTW may violate the triangle inequality (see \Cref{fig:path-angle}).

\begin{observation}
\label{obs:subsetmetric}
     Let $X$ be a finite set with distance function $\dist$. Let $Y\subset X$. Then for any $\sigma,\tau\in Y$ it holds that $\overline{\dist}(\sigma,\tau)\leq \overline{\dist|_Y}(\sigma,\tau)\leq\dist(\sigma,\tau)$.
\end{observation}
By \Cref{lem:quadineq} and \Cref{obs:subsetmetric}, $\dtwp$ on any finite set of curves in $\curvespace{m}$ is approximated by its metric closure, with approximation constant depending on $m$.
\begin{figure}
    \centering
    \includegraphics[width=\textwidth]{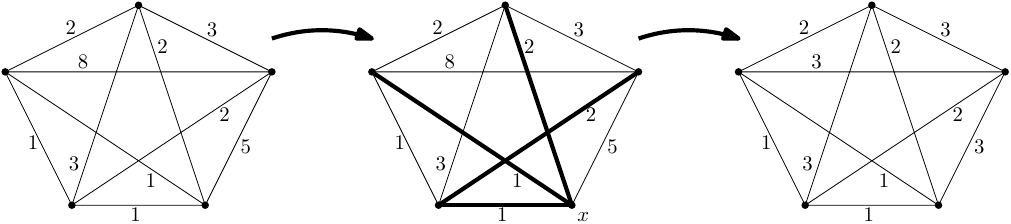}
    \caption{Illustration of the metric closure. On the left a distance function on five points represented as a graph. In the middle the shortest path tree rooted at $x$ inducing all values of the metric closure of the distance function from some element to $x$. On the right the metric closure.}
    \label{fig:metric-closure}
\end{figure}
\begin{lemma}
\label{obs:dtw_const}
    For any set of curves $X$ and two curves $\sigma,\tau\in X$ of complexity at most $m$ it holds that $\dtwp(\sigma,\tau)\leq(2m)^{1/p}\overline{\dtwp|_X}(\sigma,\tau)\leq (2m)^{1/p}\dtwp(\sigma,\tau)$.
\end{lemma}
\begin{restatable}{lemma}{medianapproximationbymetricclosure}\label{lem:bicritdtwfrommetric}
    Let $X\subset \curvespace{m}$ be a set of $n$ curves and $k$ and $\ell$ be given. Let $X^*=\{\tau^*\mid \tau\in X\}$, where $\tau^*$ is a $(1+\eps)$-approximate $\ell$-simplification of $\tau$. Let $C\subset X^*$ be an $(\alpha,\beta)$-approximation of the $k$-median problem of $X^*$ in the metric space $(X^*,\overline{\dtwp|_{X^*}})$. Then $C$ is a $\left((4m\ell)^{1/p}\left((4+2\eps)\alpha+1+\eps\right),\beta\right)$-approximation of the $(k,\ell)$-median problem on $X$.
\end{restatable}
\begin{proof}
    For any curve $\tau^*\in X^*$ let $\overline{c}(\tau^*)$ be the closest element among $C$ under the metric $\overline{\dtwp|_{X^*}}$ and for any curve $\tau$ let $c(\tau)$ be the closest element among $C$ under $\dtwp$, and let $\Delta=\sum_{\tau\in X}\dtwp(\tau,c(\tau))$ be the cost of $C$.
    Let $C^{\mathrm{opt}}=\{c_1^{\mathrm{opt}},\ldots,c_k^{\mathrm{opt}}\}$ be an optimal solution to the $(k,\ell)$-median problem on $X$ with cost $\Delta^*$.
    
    Let $V_i=\{\tau\in X\mid \forall j: \dtwp(\tau,c_i^{\mathrm{opt}})\leq \dtwp(\tau,c_j^{\mathrm{opt}})\}$ be the Voronoi cell of $c_i^{\mathrm{opt}}$, which we assume partitions $X$ by breaking ties arbitrarily. We can sassume that all sets $V_i$ are nonempty by removing the elements of $C^{\mathrm{opt}}$ with empty $V_i$. For any $i$, fix the closest $\pi_i\in V_i$ to $c_i^{\mathrm{opt}}$ under $\dtwp$. Letting $\Delta_i^*= \sum_{\sigma\in V_i}\dtw(c_i^{\mathrm{opt}},\sigma)$, we have $\Delta^*=\sum_{i\leq k}\Delta_i^*$. Let $\mathfrak{X}=X\cup X^*\cup C^\mathrm{opt}$. By \Cref{obs:subsetmetric}, for any $i\leq k$ and $\tau\in V_i$, it holds that
    \begin{align*}
        \overline{\dtwp|_\mathfrak{X}}(c_i^{\mathrm{opt}},\pi_i^*)&\leq \overline{\dtwp|_\mathfrak{X}}(c_i^{\mathrm{opt}},\pi_i)+\overline{\dtwp|_\mathfrak{X}}(\pi_i,\pi_i^*)\\
        &\leq \dtwp(c_i^{\mathrm{opt}},\pi_i)+\dtwp(\pi_i,\pi_i^*)\\
        &\leq \dtwp(c_i^{\mathrm{opt}},\pi_i)+(1+\eps)\dtwp(\pi_i,c_i^{\mathrm{opt}})\\
        &\leq (2+\eps)\dtwp(c_i^{\mathrm{opt}},\tau),
    \end{align*}
    and further observe that for
    \begin{align*}
        \overline{\dtwp|_\mathfrak{X}}(\tau^*,c_i^{\mathrm{opt}})&\leq \overline{\dtwp|_\mathfrak{X}}(\tau^*,\tau)+\overline{\dtwp|_\mathfrak{X}}(\tau,c_i^{\mathrm{opt}})\\
        &\leq \dtwp(\tau^*,\tau)+\dtwp(\tau,c_i^{\mathrm{opt}})\\
        &\leq (1+\eps)\dtwp(c_i^{\mathrm{opt}},\tau)+\dtwp(\tau,c_i^{\mathrm{opt}})\\
        &\leq (2+\eps)\dtwp(c_i^{\mathrm{opt}},\tau).
    \end{align*}
    And thus it holds that $\sum_i\sum_{\tau\in V_i}\overline{\dtwp|_X}(\tau^*,\pi_i^*)\leq (4+2\eps)\Delta^*$ (refer to \Cref{fig:vornoi-charging}). In conjunction with \Cref{lem:quadineq}, \Cref{obs:subsetmetric} and \Cref{obs:dtw_const} this yields
\begin{align*}
    \Delta&=\sum_{\tau\in X}\dtwp(\tau,c(\tau))\leq\sum_{\tau\in X}\dtwp(\tau,\overline{c}(\tau^*))\\
    &\leq(2m)^{1/p}\sum_{\tau\in X}\overline{\dtwp|_\mathfrak{X}}(\tau,\overline{c}(\tau^*))\\
    &\leq(2m)^{1/p}\sum_{\tau\in X}\left(\overline{\dtwp|_\mathfrak{X}}(\tau,\tau^*) + \overline{\dtwp|_\mathfrak{X}}(\tau^*,\overline{c}(\tau^*))\right)\\
    &\leq(2m)^{1/p}\left(\sum_{\tau\in X}\dtwp(\tau,\tau^*) + \sum_{\tau\in X}\overline{\dtwp|_{X^*}}(\tau^*,\overline{c}(\tau^*))\right)\\
    &\leq(2m)^{1/p}\left(\sum_i\sum_{\tau\in V_i}\dtwp(\tau,\tau^*) + \alpha\sum_i\sum_{\tau\in V_i}\overline{\dtwp|_{X^*}}(\tau^*,\pi_i^*)\right)\\
    &\leq(2m)^{1/p}\left(\sum_i\sum_{\tau\in V_i}\dtwp(\tau,\tau^*) + \alpha\sum_i\sum_{\tau\in V_i}\dtwp(\tau^*,\pi_i^*)\right)\\
    &\leq(2m)^{1/p}\left(\sum_i\sum_{\tau\in V_i}(1+\eps)\dtwp(\tau,c_i^{\mathrm{opt}}) + \alpha(2\ell)^{1/p}\sum_i\sum_{\tau\in V_i}\overline{\dtwp|_\mathfrak{X}}(\tau^*,\pi_i^*)\right)\\
    &\leq(2m)^{1/p}\left((1+\eps)\Delta^* + \alpha(2\ell)^{1/p}(4+2\eps)\Delta^*\right)\leq(4m\ell)^{1/p}((4+2\eps)\alpha+1+\eps)\Delta^*.\qedhere
\end{align*}
\end{proof}

\begin{restatable}{lemma}{metricclosuretime}
\label{lem:metricclosuretime}
    Let $X\subset\curvespace{\ell}$ be a set of $n$ curves. The metric closure $\overline{\dtwp|_X}$ for all pairs of curves in $X$ can be computed in $O(n^2\ell^2d+n^3)$ time.
\end{restatable}
\begin{proof}
    First compute the value $\dtwp(\sigma,\tau)$ for all pairs of curves $\sigma,\tau\in X$. This takes $O(n^2\ell^2)$ time. From this define the complete graph $\mathcal{G}(X) = (X,\binom{X}{2})$ on $X$, where the edge weights correspond to the computed values. Clearly the metric closure of $\dtwp$ corresponds to the weights of a shortest path in $\mathcal{G}(X)$. All these $\binom{n}{2}$ values can be computed in $O(n^3)$ time by $n$ applications of Dijkstra's algorithm. 
\end{proof}

\begin{theorem}[\cite{DBLP:journals/siamcomp/Chen09}]
\label{thm:metricapx}
    Given a set $P$ of $n$ points in a metric space, for $0<\eps<1$, one can compute a $(10+\eps)$-approximate $k$-median clustering of $P$ in $O(nk+k^7\eps^{-5}\log^5 n)$ time, with constant probability of success.
\end{theorem}

\begin{theorem}
\label{thm:cubic}
    Let $X$ be a set of curves of complexity at most $m$. Let $k$ and $\ell$ be given. Let $X^*=\{\tau^*\mid\tau\in X\}$ be a set of $(1+\eps)$-approximate optimal $\ell$-simplifications. There is an algorithm with input $X^*$, which computes a $(10+\eps,1)$-approximation to the $k$-median problem of $X^*$ in $(X^*,\overline{\dtwp|_{X^*}})$ in $O(n^2\ell^2d + n^3 + nk+k^7\eps^{-5}\log^5 n)$ time.
\end{theorem}
\begin{proof}
    This is a direct consequence of \Cref{lem:metricclosuretime} and \Cref{thm:metricapx}.
\end{proof}


We next show how to combine our ideas with Indyk's sampling technique for bicriteria $k$-median approximation~\cite{DBLP:conf/stoc/Indyk99} to achieve linear dependence on $n$.

\begin{figure}
    \centering
    \includegraphics[width=0.8\textwidth]{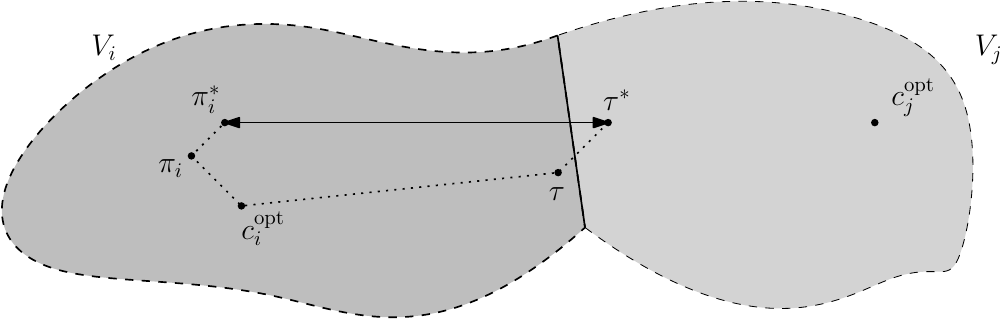}
    \caption{Illustration to Proof of \Cref{lem:bicritdtwfrommetric}: Assigning $\tau^*$ (the $(1+\eps)$-simplification of $\tau$ which lies inside the Voronoi cell $V_i$ of $c_i\opt$) to $\pi_i^*$ (the $(1+\eps)$-simplification of the closest element $\pi_i$ in $V_i$ to $c_i^{\mathrm{opt}}$) under $\overline{\dtwp|_{X^*}}$ is at most $4+2\eps$ times as bad as assigning $\tau$ to $c_i^{\mathrm{opt}}$ under $\dtwp$.}
    \label{fig:vornoi-charging}
\end{figure}

\subsection{Linear time algorithm}

With \Cref{thm:cubic} (and by extension also \Cref{cor:cubic62apx}) we have ran into the following predicament: We would like to apply linear time algorithms to the metric closure of $\dtwp$. However, constructing the metric closure takes cubic time. We circumvent this by applying the following algorithm, which reduces a $k$-median instance with $n$ points to two $k$-median instances with $O(\sqrt{n})$ points, simply by sampling. More precisely, we will apply this technique twice, so that we will compute the metric closure only on sampled subsets of size $O(n^{1/4})$.  
In this section we want to analyse the problem of computing a $k$-median of a set $X$ in the metric space $(X,\overline{\dist})$, where $\dist$ is a distance function on $X$ with the guarantee that there is a constant $\distconst$ such that for any $x,y\in X$ it holds that $\dist(x,y)\leq \distconst\overline{\dist}(x,y)$, with a linear running time, and more precisely, only a linear number of calls to the distance function $\dist$, and no calls to $\overline{\dist}$. By \Cref{obs:dtw_const} the results in this section translate directly to $\dist=\dtwp|_X$ with $\distconst = (m+\ell)^{1/p}$.

Observe, that similar to \Cref{thm:cubic}, the following lemma holds.

\begin{lemma}
\label{lem:general_cubic}
    Let $X$ be a set of $n$ points, equipped with a distance function $\dist$ that can be computed in time $T_\dist$. There is a $(10+\eps,1)$-approximate algorithm for $k$-median of $X$ in $(X,\overline{\dist})$ that has constant probability of success and has running time $O(n^2T_\dist+n^3 + nk + k^7\eps^{-5}\log^5n)$.
\end{lemma}

\begin{restatable}{lemma}{subsetmetricguarantee}
\label{lem:subsetmetricguarantee}
     Let $X$ be a set of $n$ points, equipped with a distance function $\dist$, such that $\dist\leq \distconst\overline{\dist}$ for some $\distconst>0$, and $Y\subset X$. A $(\alpha,\beta)$-approximation for the $k$-median problem for $Y$ in $(Y,\overline{\dist|_Y})$ is a $(\alpha \distconst,\beta)$-approximation for the $k$-median problem for $Y$ in $(Y,\overline{\dist}|_Y)$.
\end{restatable}

\begin{proof}
    Let $C\subset Y$ be a $(\alpha,\beta)$-approximation for the $k$-median problem for $Y$ in $(Y,\overline{\dist|_Y})$, and let $C\opt=\{c_1\opt,\ldots,c_k\opt\}\subset Y$ be an optimal solution for the $k$-median problem for $Y$ in $(Y,\overline{\dist}|_Y)$ with cost $\Delta\!\opt$. For any $\tau\in Y$ let $c_Y(\tau)$ be the closest element among $C$ under $\overline{\dist|_Y}$, and let $c\opt_X(\tau)$ be the closest element among $C\opt$ under $\overline{\dist}=\overline{\dist}|_Y$. Then $\Delta\!\opt=\sum_{\tau\in Y}\overline{\dist}(\tau,c\opt_X(\tau))$. Overall by \Cref{obs:subsetmetric} we see that
    \begin{align*}
        \sum_{\tau\in Y}\overline{\dist}(\tau,c_Y(\tau))&\leq \sum_{\tau\in Y}\overline{\dist|_Y}(\tau,c_Y(\tau))\leq\alpha\sum_{\tau\in Y}\overline{\dist|_Y}(\tau,c_X(\tau))\\
        &\leq\alpha\sum_{\tau\in Y}\dist(\tau,c_X(\tau))\leq\alpha\sum_{\tau\in Y}\distconst\overline{\dist}(\tau,c_X(\tau))=\distconst\alpha\Delta\!\opt.\qedhere
    \end{align*}
\end{proof}


\begin{algorithm}
\caption{$k$-median framework}\label{alg:iterativekmedianalgorithm2}
\begin{algorithmic}
\Procedure{$k$-Routine}{$(X,\dist),\eps,\mathcal{A}$}\Comment{$\mathcal{A}$ is $(\alpha,\beta)$-approximate metric $k$-median}
\State $a\gets \Theta(\eps^{-1}\sqrt{\log(\eps^{-1})})$, $b\gets \Theta(a^2)$\Comment{Determine the success probabilities.}
\State $s\gets a\sqrt{kn\log k}$
\State Choose a set $S$ of $s$ points sampled without replacement from $X$
\State $C'\gets\mathcal{A}((S,\overline{\dist|_S}))$
\State Select the set $M$ of points $x$ with the $b\frac{kn\log k}{s}$ largest values of $\min_{c'\in C'}\overline{\dist}(x,c')$
\State \textbf{return} $C=C'\cup \mathcal{A}((M,\overline{\dist|_M}))$
\EndProcedure

\Procedure{$k$-Median}{$(X,\dist),\eps,\mathcal{A}$}\Comment{$\mathcal{A}$ is $(\alpha,\beta)$-approximate metric $k$-median}
\State $a\gets \Theta(\eps^{-1}\sqrt{\log(\eps^{-1})})$, $b\gets \Theta(a^2)$\Comment{Determine the success probabilities.}
\State $s\gets a\sqrt{kn\log k}$
\State Choose a set $S$ of $s$ points sampled without replacement from $X$
\State $C'\gets\text{$k$-\textsc{Routine}}((S,\dist|_S),\eps,\mathcal{A})$
\State Select the set $M$ of points $x$ with the $b\frac{kn\log k}{s}$ largest values of $\min_{c'\in C'}\dist(x,c')$
\State \textbf{return} $C=C'\cup \text{$k$-\textsc{Routine}}((M,\dist|_M),\eps,\mathcal{A})$
\EndProcedure
\end{algorithmic}
\end{algorithm}

\begin{theorem}[\cite{DBLP:conf/stoc/Indyk99}]
\label{thm:iterativekmed}
    Let $\mathcal{A}$ be a $(\alpha,\beta)$-approximate algorithm for $k$-median in metric spaces with constant success probability. Then for any $\eps>0$ the $k$-\textsc{Routine} in \Cref{alg:iterativekmedianalgorithm2} provided with $\mathcal{A}$ is a $(3(1+\eps)(2+\alpha),2\beta)$-approximate algorithm for $k$-median in metric spaces with constant success probability.
\end{theorem}

\begin{restatable}{lemma}{iterativekmedrunningtime}
    \label{lem:iterativekmedrunningtime}
    Let $X$ be a set of $n$ points, and let $\dist$ be a distance function that can be computed in $T_\dist$ time for any $x,y\in X$. Let $T_\mathcal{A}(n)$ be the running time of the $(\alpha,\beta)$-approximate algorithm for $k$-median on $n$ elements. Then $k$-\textsc{Routine} has a running time of $O(n^2T_\dist + T_\mathcal{A}(\min(n,\eps^{-1}\sqrt{kn\log(k)\log(\eps^{-1})})))$.
\end{restatable}

\begin{proof}
    The only steps that take time outside the two calls to $\mathcal{A}$ are sampling $S$ and construcing $M$. Computing the values $\min_{c'\in C'}\overline{\dist}(x,c')$ for all $x\in X$ can be done in a single execution of Dijkstra's algorithm, starting with the points of $C'\subset X$ at distance~$0$, by adding a temporary point with distance $0$ to all points in $C'$ and starting Dijkstra's algorithm on this temporary point. This takes $O(n^2T_\dist)$ time, which also dominates the time it takes to sample $S$ as well as constructing $M$ from these computed values.
\end{proof}

\begin{lemma}
\label{thm:subroutine}
    Let $X$ be a set of $n$ points, and let $\dist$ be a distance function on $X$, which can be computed in time $T_\dist$, and further there is a constant $\distconst$ such that $\dist\leq \distconst\overline{\dist}$. Let $Y\subset X$. Let $\eps>0$ and let $\mathcal{A}$ be the $(10+\eps,1)$-approximation for metric $k$-median of \Cref{lem:general_cubic}. Then $k$-\textsc{Routine} returns a $(3(1+\eps)\distconst(12+\eps),2)$-approximation of $k$-median in the metric space $(Y,\overline{\dist}|_Y)$ in time $O(|Y|^2T_\dist + |Y|^2k\log(k)\eps^{-2}\log(\eps^{-1}) + k^7\eps^{-5}\log^5(|Y|)))$.
\end{lemma}
\begin{proof}
    The running time bound follows by \Cref{lem:iterativekmedrunningtime} and \Cref{lem:general_cubic}, together with the fact, that $\min(|Y|,\eps^{-1}\sqrt{k|Y|\log(k)\log(\eps^{-1})})^{3}\leq|Y|^2k\log(k)\eps^{-2}\log(\eps^{-1})$.
    The approximation guarantee follows by \Cref{thm:iterativekmed}, \Cref{lem:subsetmetricguarantee} and \Cref{lem:general_cubic}.
\end{proof}

By combining the presented subroutines, we obtain our two main results of the section. The first is \Cref{lem:main}, which provides a linear time approximation algorithm for $k$-median in metric closures, assuming the underlying distance is reasonably well approximated by its metric closure. The second is \Cref{cor:second_bicriterial}, combining \Cref{lem:main} with \Cref{lem:bicritdtwfrommetric} to yield an approximation algorithm for $p$-DTW with an unoptimized approximation guarantee.

\begin{lemma}[\cite{DBLP:conf/stoc/Indyk99}]
\label{lem:largeset}
    Assume that $C'$, computed in the algorithm $k$-\textsc{Median}, is an $(\alpha,\beta)$-approximation of $k$-median of $S$ in the metric space $(S,\overline{\dist}|_S)$. With constant probability depending only on $a$ and $b$, there is a subset of $X$ of size at least $n-b\frac{kn}{s}\log k$, whose cost under $\overline{\dist}$ with $C'$ as medians is at most $(1+\eps)(2+\alpha)\Delta\!\opt$, where $\Delta\!\opt$ is the cost of an optimal $k$-median under $\overline{\dist}$ of $X$.
\end{lemma}

\begin{restatable}{theorem}{linearlemma}
\label{lem:main}
    Let $X$ be a set of points and let $\dist$ be a distance function on $X$ with $\dist\leq\distconst\overline{\dist}$. Let $\eps>0$ and let $\mathcal{A}$ be the $(10+\eps,1)$-approximation for metric $k$-median of \Cref{thm:cubic}. Then $k$-\textsc{Median} returns a  $(11\distconst^2(1+\eps)^2(12+\eps),4)$-approximation of $k$-median of $X$ in the metric space $(X,\overline{\dist})$ in time $O(nk\log(k)T_\dist + nk^{2}\log^{2}k + k^7\eps^{-5}\log^5(n))$.
\end{restatable}
\begin{proof}
    Let $\Delta\!\opt$ be the cost of an optimal solution to the $k$-median problem on $X$ under $\phi$.
    By \Cref{thm:subroutine} the set $C'$ is a $(3(1+\eps)\distconst(12+\eps),2)$-approximation of $k$ median of $S$ in $(S,\overline{\dist}|_S)$ and can be computed in time $O(nk\log(k)T_\dist + nk^{2}\log^{2}(k)\eps^{-4}\log^2(\eps^{-1}) + k^7\eps^{-5}\log^5(n))$. The set $M$ can be computed in $O(knT_\dist)$ time. By \Cref{thm:subroutine} the set $C''$ is a $(3(1+\eps)\distconst(12+\eps),2)$-approximation of $k$ median of $M$ in $(M,\overline{\dist}|_M)$ and can be computed in time $O(nk\log(k)T_\dist + nk^{2}\log^{2}(k)\eps^{-4}\log^2(\eps^{-1}) + k^7\eps^{-5}\log^5(nk))$.

    Now observe that for the set $O$ of \Cref{lem:largeset},
    \begin{align*}
        \sum_{x\in X\setminus M}\overline{\dist}(x,C')&\leq \sum_{x\in X\setminus M}\dist(x,C')\leq \sum_{x\in O}\dist(x,C')\\
        &\leq \sum_{x\in O}\distconst\overline{\dist}(x,C')\leq \distconst(1+\eps)(2+3(1+\eps)\distconst(12+\eps))\Delta\!\opt.
    \end{align*}
    Further observe, that there exists a clustering of $M$ with cost $2\Delta\!\opt$ by replacing every point of an optimum with its closest point in $M$. Thus,
    \begin{align*}
        \sum_{x\in M}\overline{\dist}(x,C'')=\sum_{x\in M}\overline{\dist}|_M(x,C'')\leq 6(1+\eps)\distconst(12+\eps)\Delta\!\opt.
    \end{align*}
    Overall we get a $(11\distconst^2(1+\eps)^2(12+\eps),4)$-approximation (as $\distconst\geq 1$ and $\eps>0$) in time $O(nk\log(k)T_\dist + nk^{2}\log^{2}(k)\eps^{-4}\log^{2}(\eps^{-1}) + k^7\eps^{-5}\log^5(n))$.
\end{proof}

We briefly discuss simplification schemes for curves under $p$-DTW (for more details refer to \Cref{appsec:simplifications}). 
We reduce the problem of finding an $(1+\eps)$-approximate simplification to finding a $(1+\eps)$-approximation of a center point for a set of $\leq m$ points, where the objective is to minimize the sum of the individual distances to the center point raised to the $p$th power. Note that for $p=\infty$, the problem is that of finding a minimum enclosing ball, and for $p=2$, the problem can be reduced to that of finding the center of gravity of the set of discrete points, which can both be solved exactly. Furthermore, we show (\cref{prop:naivesimplification}) that for all $\dtwp$, there is a deterministic $2$-approximation that is a crucial ingredient for our approximation algorithms of $(k,\ell)$-median under $p$-DTW. 

\begin{restatable}{proposition}{naivesimplification}
\label{prop:naivesimplification}
    For $\sigma=(\sigma_1,\ldots,\sigma_m)\in \curvespace{m}$ and integer $\ell>0$, one can compute in $O(m^2(d+\ell+m))$ time a curve $\sigma^*\in\curvespace{\ell}$ such that
    \[\inf_{\sigma_{\ell}\in \curvespace{\ell}}\dtwp(\sigma_{\ell},\sigma)\leq\dtwp(\sigma^*,\sigma)\leq2\inf_{\sigma_{\ell}\in \curvespace{\ell}}\dtwp(\sigma_{\ell},\sigma).\]
\end{restatable}

\begin{corollary}
\label{cor:second_bicriterial}
    For any $\eps>0$ the procedure $k$-\textsc{Median} from \Cref{alg:iterativekmedianalgorithm2} can be used to compute a $(72(1+\eps)^2(12+\eps)(16m\ell^3)^{1/p},4)$-approximation for $(k,\ell)$-median for an input set $X$ of $n$ curves of complexity $m$ under $\dtwp$ in time $O(nm^3d + nk\log(k)\ell^2d + nk^{2}\log^{2}(k)\eps^{-4}\log^2(\eps^{-1}) + k^7\eps^{-5}\log^5(n))$.
\end{corollary}
\begin{proof}
    Let $X^*=\{\tau^*\mid\tau\in X\}$ be a set of $2$-approximate optimal $\ell$-simplifications of $X$ under $\dtwp$. By \Cref{prop:naivesimplification}, $X^*$ can be computed in $O(nm^3d)$ time.
    We now apply \Cref{lem:main} and \Cref{obs:dtw_const} to obtain a $(12(2\ell)^{2/p}(1+\eps)^2(12+\eps),4)$-approximation of $k$-median of $X^*$ in $(X^*,\overline{\dtwp|_{X^*}})$ in time $O(nk\log(k)\ell^2d + nk^{2}\log^{2}(k)\eps^{-4}\log^2(\eps^{-1}) + k^7\eps^{-5}\log^5(n))$. By \Cref{lem:bicritdtwfrommetric}, the computed set is a $(6(4m\ell)^{1/p}12(2\ell)^{2/p}(1+\eps)^2(12+\eps),4)$-approximation for $(k,\ell)$-median for $X$ under $\dtwp$.
\end{proof}

\section{Coreset Application}\label{sec:constructingcoresets}

  The theoretical derivations of the previous sections culminate in an approximation algorithm (Theorem~\ref{thm:mainapplication}) to $(k,\ell)$-median that is particularly useful in the big data setting, where $n\gg m$. Our strategy is to first compute an efficient but not very accurate approximation(\Cref{cor:second_bicriterial}) of $(k,\ell)$-median. Subsequently, we use the approximation to construct a coreset. The coreset is then investigated using its metric closure, where by virtue of the size reduction we can greatly reduce the running time of slower more accurate algorithms metric approximation algorithms, yielding a better approximation.

\begin{theorem}[\cite{k_median_local_search,DBLP:journals/siamcomp/Chen09}]
\label{thm:metric5apx}
    Given a set $X$ of $n$ points in a metric space, one can compute a $(5+\eps)$-approximate $k$-median clustering of $X$ in $O(\eps^{-1}n^2k^3\log n)$ time. If $P$ is a weighted point set, with total weight $W$, then the time required is in $O(\eps^{-1}n^2k^3\log W)$.
\end{theorem}

\begin{algorithm}
\caption{$((32+\eps)(4m\ell)^{1/p})$-approximate $(k,\ell)$-median}\label{alg:finalalg}
\begin{algorithmic}
\Procedure{$(k,\ell)$-Median}{$X\subset\curvespace{m},p,\eps$}
\State $\eps'\gets\eps/46$
\State Compute $(O((16m\ell^3)^{1/p}),4)$-approximation $C'$ (\Cref{cor:second_bicriterial})
\State Compute bound of sensitivity for each curve $x\in X$ from $C'$ (\Cref{lem:epssensitivities})
\State Compute sample size $s\gets O(\eps^{-2}d\ell k^2(m^2\ell^4)^{1/p}\log^3(m\ell)\log^2(k)\log(\eps ^{-1})\log(n))$
\State Sample and weigh $\eps'$-coreset $S$ of $X$ of size $s$ (\Cref{theo:coreset})
\State Compute a $2$-simplification for every $s\in S$ resulting in the set $S^*$ (\Cref{prop:naivesimplification})
\State Compute metric closure values $\overline{\dist} = \overline{\dtwp|_{S^*}}$ (\Cref{lem:metricclosuretime})
\State \textbf{Return} $(5+\eps',1)$-approximation of weighted $k$-median in $(S^*,\overline{\dist})$ (\Cref{thm:metric5apx})
\EndProcedure
\end{algorithmic}
\end{algorithm}

\begin{restatable}{theorem}{mainlinearthm}\label{thm:mainapplication}
    Let $0<\eps\leq1$. The algorithm $(k,\ell)$-\textsc{Median} in \Cref{alg:finalalg} is a $((32+\eps)(4m\ell)^{1/p},1)$-approximate algorithm of constant success probability for $(k,\ell)$-median on curves under $\dtwp$ with a running time of
    \(\widetilde{O}\!\left(n(m^3d+k^2+k\ell^2d) + \eps^{-6}d^3\ell^3k^7\sqrt[p]{m^6\ell^{12}}\right)\),
    where $\widetilde{O}$ hides polylogarithmic factors in $n$, $m$, $\ell$, $k$ and $\eps^{-1}$.
\end{restatable}

\begin{proof}
Computing a $(O(m^{1/p}\ell^{3/p}),4)$-approximation via \Cref{cor:second_bicriterial} takes time $O(nm^3d + nk\log(k)\ell^2 + nk^{2}\log^{2}(k) + k^7\log^5(n))$ and has constant success probability. From this we can compute a $\eps'$-coreset $S$ by \Cref{theo:coreset} of size \[O(\eps ^{-2}d k^2(m^2\ell^4)^{1/p}\log^3(m\ell)\log^2(k)\log(\eps ^{-1})\log(n))\] with constant success probability, and in time $O(kn)$. Computing $S^*$ takes $O(nm^3)$ time. Computing the metric closure of $S^*$ takes $O(|S|^3)$ time and computing a $(5+\eps')$-approximate solution to the $k$-median solution of $S^*$ takes $O(\eps^{-1}|S|^2k^3\log|S|)$ time by \Cref{thm:metric5apx}. This is a $(4m\ell)^{1/p}(32+13\eps')(1+\eps')$-approximation to $(k,\ell)$-median of $X$ by \Cref{lem:bicritdtwfrommetric} and \Cref{theo:coreset}. Overall the approximation factor is $(4m\ell)^{1/p}(32+\eps)$ and the running time is in 
\[O\!\left(\!n\!\left(m^3d+k\log k\ell^2 + (k\log k)^2\right)\! + \eps^{-6}d^3\ell^3k^7\sqrt[p]{m^6\ell^{12}}\log^9(m\ell)\log^6(k)\log^5(n)\log^3\!\left(\frac{1}{\eps}\right)\!\!\right)\!.\qedhere\]
\end{proof}

Combining the computed $\eps$-coreset with the $(k,\ell)$-median algorithm from~\cite[Theorem~35]{DBLP:conf/waoa/BuchinDGPR22} instead, we achieve a matching approximation guarantee and improve the dependency on $n$. The improved approximation guarantee from \cref{cor:koen} compared to \Cref{thm:mainapplication} comes at the cost of an exponential dependency in $k$, as is also present in their results.
\begin{corollary}
\label{cor:koen}
    Let $0<\eps\leq 1$ and $0<\delta\leq 1$. There is an $((8+\eps)(m\ell)^{1/p},1)$-approximation for $(k,\ell)$-median with $\Theta(1-\delta)$ success probability and running time in 
    \[\widetilde{O}\left(n(m^3d+k^2+k\ell^2) + k^7 + \left(32k^2\eps^{-1}\log(1/\delta)\right)^{k+2}md\left(m^3+\eps^{-2}d\ell k^2\sqrt[p]{m^2\ell^4}\right)\right),\]
    where $\widetilde{O}$ hides polylogarithmic factors in $n$, $m$, $\ell$, $k$ and $\eps^{-1}$.
\end{corollary}
Finally, combining \Cref{thm:mainapplication} with \Cref{theo:coreset} yields the following result.
\begin{corollary}
\label{cor:maincoreset}
    The algorithm $(k,\ell)$-\textsc{Median} in \Cref{alg:finalalg} can be used to construct an $\eps$-coreset for $(k,\ell)$-median in time \(\widetilde{O}\!\left(n(m^3d+k^2+k\ell^2d) + \eps^{-6}d^3\ell^3k^7\sqrt[p]{m^6\ell^{12}}\right)\) with constant success probability of size 
    \[O(\eps^{-2}d\ell k^2(m^2\ell^2)^{1/p}\log^3(m\ell)\log^2(k)\log(\eps ^{-1})\log(n)).\]
\end{corollary}


\section{Conclusion}

Our first contribution involves investigating the VC dimension of range spaces characterized by arbitrarily small perturbations of DTW distances. While our results hold for a relaxed variant of the range spaces in question, they establish a robust link between numerous sampling results dependent on the VC dimension and DTW distances. Indeed, our first algorithmic contribution is the construction of coresets for $(k,\ell)$-median through the sensitivity sampling framework by Feldman and Langberg~\cite{Feldman2011}. Apart from the VC dimension, the crux of adapting the sensitivity sampling framework to our (non-metric) setting was to use an already known weak version of the triangle inequality satisfied by DTW. This inequality prompted us to further explore approximation algorithms by approximating DTW with a metric. 
By reducing to the metric case and plugging in our coresets, we designed an algorithm for the $(k,\ell)$-median problem, with running time linear in the number of the input sequences, and an approximation factor predominantly determined by our generalised iterated triangle inequality.

Although our primary motivation lies in constructing coresets, there are additional direct consequences through sampling bounds that establish a connection between the sample size and the VC dimension. For instance, suppose that we have a large set of time series, following some unknown distribution, and we want to estimate the probability that a new time series falls within a given DTW ball $b$. Suppose that we also allow for small perturbations of the distances, i.e., we only want to guarantee that the estimated probability is realized by \emph{some} small perturbations of the distances. This probability can be approximated within a constant additive error, by considering a random sample of size depending solely on the VC dimension and the probability of success (over the random sampling) and measuring its intersection with $b$ (see e.g. \Cref{theo:etaepsilonapprox}). 
 Such an estimation can be used for example in anomaly detection, where one aims to detect time series with a small chance of occurring, or in time series segmentation,  where diverse patterns may emerge throughout the series.

\bibliographystyle{plainurl}
\bibliography{biblio.bib}
\appendix

\section{VC dimension analysis}
\label{sec:vcproofs}
\begin{restatable}{lemma}{vcexact}\label{lem:VCDTWevenp}
    If $p$ is even, then the VC dimension of $\left(\curvespace{=m},\left\{B_{r,m}^p(\sigma)\mid \sigma \in \bX_{=\ell}^d\right\} \right)$ is in $O\left( d\ell^2 \log (m p) \right)$. 
\end{restatable}
\begin{proof}
For any two vectors $x=(x_1,\ldots,x_d), y =( y_1,\ldots,y_d)$ let $x\oplus y \in \RR^{2d}$ be their concatenation $(x_1,\ldots,x_d, y_1,\ldots,y_d)$. 
    We apply \cref{thm:vcpoly} as follows. The parameter vector $\alpha(\sigma, r)$ encodes the center $\sigma$ and the radius $r$ of each ball. More formally, for any ball center $\sigma = (\sigma_1, \ldots,\sigma_{\ell})$ and radius $r$, we define a polynomial function parameterized by $\alpha(\sigma, r) = \sigma_1 \oplus \dots \oplus \sigma_{\ell} \oplus (r) \in \RR^{d\ell+1}$. Let $F$ be the family of at most $m^{O(\ell)}$ 
    polynomial functions, each one realizing the difference between the $p$\textsuperscript{th} power of the radius and the $p$\textsuperscript{th} power of the cost of a different traversal. More formally, for a curve $\tau = (\tau_1, \ldots, \tau_m) \in \bX_{m}^d$:
    \[
    F = \left\{ \alpha(\sigma, r), \tau_1 \oplus \dots \oplus \tau_m \mapsto r^p - \sum_{(i,j)\in T}  \| \sigma_i - \tau_j\|_2^p    \mid \sigma \in \bX_{\ell}^d, r\in \RR_{+}, T \in \traversals_{\ell, m}  \right\}. 
    \]
    Since $p$ is even, the square root in $\|\cdot \|_2$ vanishes and the degree is upper bounded by $p$. Finally, we define a boolean function $g:~\{-1,1\}^{|F|} \to \{0,1\}$ which simply returns $1$ iff at least one of the arguments is non-negative. The function $g$ defines $H = (X,\RRR_{m,\ell}^p)$ as a $m^{O(\ell)}$-combination of $\sign(F)$, and by \cref{thm:vcpoly}, we conclude that the VC dimension of $H$ are in $O\left( d\ell^2 \log (m p) \right)$.
\end{proof}

\section{Approximating the DTW distance by a metric}\label{sec:metricapproximation}

\begin{corollary}
    Let $X$ be a set of curves of complexity at most $m$. Let $k$ and $\ell$ be given. There is an algorithm which computes a $((4m\ell)^{1/p}(41+O(\eps)),1)$-approximation to the $(k,\ell)$-median problem on $X$ under $\dtw$ in $O(nm^3\ell\log^3(m\eps^{-1}) + n^2\ell^2 + n^3 + nk+k^7\eps^{-5}\log^5 n)$.
\end{corollary}
\begin{proof}
    This is a direct consequence of \Cref{thm:cubic}, \Cref{lem:bicritdtwfrommetric} and \Cref{thm:simplification}.
\end{proof}

\begin{corollary}
\label{cor:cubic62apx}
    Let $X$ be a set of curves of complexity at most $m$. Let $k$ and $\ell$ be given. There is an algorithm which computes a $((4m\ell)^{1/p}(62+O(\eps)),1)$-approximation to the $(k,\ell)$-median problem on $X$ under $\dtwp$ in $O(nm^3 + n^2\ell^2 + n^3 + nk+k^7\eps^{-5}\log^5 n)$.
\end{corollary}
\begin{proof}
    This is a direct consequence of \Cref{thm:cubic}, \Cref{lem:bicritdtwfrommetric} and \Cref{prop:naivesimplification}.
\end{proof}

\section{Approximate $\ell$-simplifications under DTW}\label{appsec:simplifications}

\begin{definition}
With the notation of \Cref{def:pq_dtw}, we call a traversal \emph{lopsided} (refer to \Cref{fig:lopsidedtraversal}) if for all $i$ we have that $(a_{i+1},b_{i+1})-(a_i,b_i)\in\{(0,1),(1,1)\}$.
\end{definition}
\begin{figure}
    \centering
    \includegraphics[width=\textwidth]{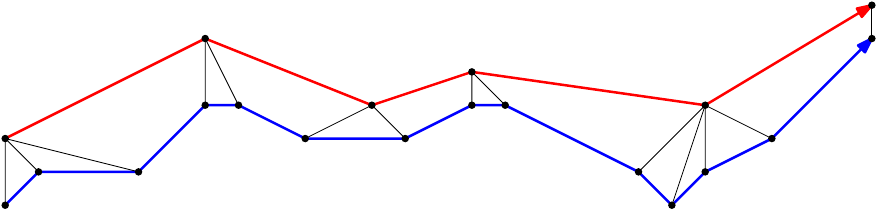}
    \caption{Illustration of a lopsided traversal.}
    \label{fig:lopsidedtraversal}
\end{figure}
\begin{lemma}\label{lem:nosteppysteppy}
Let $\sigma\in\curvespace{m}$ and $\tau\in\curvespace{n}$. There is an optimal traversal $T^*\in T_{n,m}$ that realizes
$\dtw(\sigma,\tau)$
and for all $1\leq i\leq |T^*|-2$ it holds that $(a_{i+2},b_{i+2})-(a_i,b_i)\neq(1,1)$.
\end{lemma}
\begin{proof} 
Let $T$ be an arbitrary traversal for which there is an index $i$ such that $(a_{i+2},b_{i+2})-(a_i,b_i)=(1,1)$. Observe that we can remove the index pair $(a_{i+1},b_{i+1})$ from the traversal without increasing the distance, which implies the claim.
\end{proof}
\begin{lemma}\label{lem:infimumlpt}
Let $\tau\in\curvespace{m}$ and $\sigma\in\curvespace{\ell}$ satisfy  
\begin{align}\label{eq:minimizer}
\dtwp(\sigma,\tau) = \inf_{\sigma'\in \curvespace{\ell}}\dtwp(\sigma',\tau).
\end{align}
Then $\dtwp(\sigma,\tau)$ can be realized with a lopsided traversal.
\end{lemma}
\begin{proof}

    Refer to \Cref{fig:leftpalmify}.
    Assume that $\dtwp(\sigma,\tau)$ is not attained by a lopsided traversal. It suffices to show that there is then a curve $\sigma'\in\curvespace{\ell}$ with $\dtwp(\sigma',\tau) \leq \dtw(\sigma,\tau)$.
    
    Fix some traversal. There is an index $i$ such that $(a_{i+1},b_{i+1})-(a_i,b_i) = (1,0)$. Observe that removing vertex $\sigma_{a_{i+1}}$ from $\sigma$ does not increase $\dtwp(\sigma,\tau)$. Indeed, by \Cref{lem:nosteppysteppy}, we can assume that $(a_{i+2},b_{i+2})-(a_{i+1},b_{i+1}) \neq (0,1)$, so $a_{i+2} = a_{i+1}+1$. Therefore, discarding $\sigma_{a_{i+1}}$ removes exactly one summand from the sum in $\dtw(\sigma,\tau)$, yielding a contradiction.
\end{proof}
\begin{figure}
    \centering
    \includegraphics[width=\textwidth]{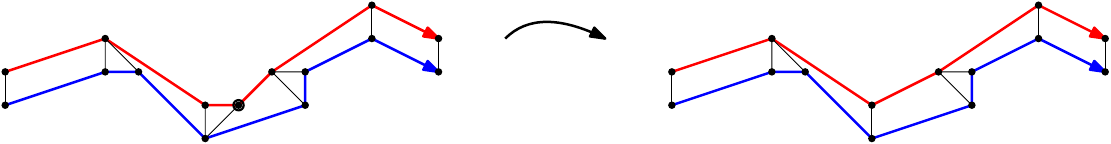}
    \caption{Illustration of Proof of \Cref{lem:infimumlpt}. The vertex that can safely be deleted is marked on the left, and removed on the right.}
    \label{fig:leftpalmify}
\end{figure}
By the following theorem, we can calculate a $(1+\eps)$-approximate $\ell$-simplification for any $\eps$ for the classical DTW distance $\dtw:=\dtw_1$ with some success probability. 
\begin{theorem}\label{thm:simplification}
There is a constant $\delta<1$ $\sigma\in\curvespace{m}$, $\ell\in\NN$ and $\eps>0$, one can compute in $O(\ell m^2+m^3d\log^3(m\eps^{-1}))$ time a curve $\sigma^*\in \curvespace{\ell}$ such that
$$
\inf_{\sigma_{\ell}\in \curvespace{\ell}}\dtw(\sigma_{\ell},\sigma)\leq \dtw(\sigma^*,\sigma)\leq (1+\eps)\inf_{\sigma_{\ell}\in \curvespace{\ell}}\dtw(\sigma_{\ell},\sigma),$$
with success probability $(1-\delta)^\ell$.
\end{theorem}
\begin{proof}
    If $m\leq \ell$, return $\sigma$.
    If $m>\ell$, the idea is to use \Cref{lem:infimumlpt} and find a curve in $\curvespace{\ell}$ along with a lopsided traversal such that the inequality holds.

    If $\sigma^* \in \curvespace{\ell}$ is such that $\dtw(\sigma^*,\sigma)$ is minimal over all curves in $\curvespace{\ell}$, by \Cref{lem:infimumlpt}, there exists a lopsided traversal realizing the distance. Observe that the lopsided traversal induces a partitioning of the vertices of $\sigma$ into $\ell$ distinct sets $G_1,\ldots,G_{\ell}$ of vertices. Each $G_i$ contains a list of consecutive vertices of $\sigma$ that connect to a given vertex of $\sigma^*=(\sigma^*_1,\ldots,\sigma^*_{\ell})$ and we have that 
    \[
    \dtw(\sigma^*,\sigma) = \sum_{i=1}^{\ell}\sum_{p\in G_i}\|\sigma_i-p\|.
    \]
    Conversely, any partitioning of the vertices of $\sigma$ into $\ell$ distinct sets of consecutive vertices together with a point in $\RR^d$ for each of these sets gives rise to a curve in $\curvespace{\ell}$. 
    We thus see that it is enough to find such a partition along with a vertex for every set in the partition, such that the inequality holds.
    
    This problem can be solved similarly to one dimensional $\ell$-medians clustering. That is, for $\sigma=(\sigma_1,\ldots,\sigma_{m})$, we define a cost function $C_{\sigma}(a,b)=\inf_{c\in \RR^d}\sum_{i=a}^b\|\sigma_i-c\|$, corresponding to the cost of grouping indices $\sigma_a,\ldots,\sigma_b$, with some optimal point in $\bR^d$. It is well-known that this cost function can be $(1+\eps)$-approximated in $O(md\log^3(m\eps^{-1}))$ time \cite{GeometricMedian16} with constant success probability $1-\delta$. Hence, it can be approximately computed in $O(m^3d\log^3(m\eps^{-1}))$ time for all index pairs $a$ and $b$, such that the algorithm succeeds in approximating the cost for an optimal partition with probability $(1-\delta)^\ell$, where $\delta<1$ is the constant success probability of the algorithm presented in \cite{GeometricMedian16}. We next compute $D(i,j)$, the approximate minimal cost to group the vertices $\sigma_1,\ldots,\sigma_i$ into $j$ groups for all $i\leq m$, $j\leq \ell$. This can be done in $O(\ell m^2)$ time, by first setting $D(i,1)=C(1,i)$ and then computing $D(i,j)=\min_{l\leq i}D(l,j-1) + C(l+1,i)$. Overall, it thus takes $O(\ell m^2 + m^3d\log^3(m\eps^{-1}))$ time to $(1+\eps)$-approximately compute an optimal curve with success probability $1-\delta^\ell$.
\end{proof}

\begin{corollary}\label{cor:simplify}
For $\sigma\in\curvespace{m}$ and $m\geq \ell>0$, one can compute in $O(\ell m^2 + m^3)$ time a curve $\sigma'$ with vertices a subset of vertices in $\sigma$, such that

\[
\dtw(\sigma',\sigma)=\inf_{\sigma_{\ell}}\dtw(\sigma_{\ell},\sigma),
\]
where the infimum is taken over all $\sigma_{\ell}\in\curvespace{\ell}$ with vertices lying on vertices of $\sigma$. 
\end{corollary}
\begin{proof}
This follows immediately from the fact that the cost  $C_P(a,b)$ of clustering the $a$-th to $b$-th vertex of $\sigma$ with a single vertex of $\sigma$ can be computed in $O(m(b-a))=O(m^2)$ time.
\end{proof}

\subsection{Approximate $\ell$-simplifications under $\dtwp$}
The same ideas as above lead to a deterministic $2$-approximation for $\dtwp$ for arbitrary values of $p$.
\naivesimplification*
\begin{proof}
    First compute in $O(m^2d)$ time all pairwise distances of vertices of $\sigma$, and store them for later use. Next compute $C(a,b)=\min_{a\leq i\leq b}\sum_{a\leq j\leq b}\|\sigma_i-\sigma_j\|_2^p$ for all $1\leq a\leq b\leq m$. This can be done in $O(m^3)$ time, as for any fixed $1\leq a\leq b\leq m$, the computations of $C(a,b)$ are a big part of the computations of $C(a,b+1)$. To be more precise, when computing $C(a,b)$ we store $O(m)$ temporary values $\nu(i)=\sum_{a\leq j\leq b}\|\sigma_i-\sigma_j\|_2^p$ for all $a\leq i\leq b$. Then $C(a,b)=\min_{a\leq i\leq b}\nu(i)$. Once we have these values, $C(a,b+1)$ can be computed in $O(m)$ time, by updating $\nu(i)$ with $\nu(i) + \|p_i-p_{b+1}\|_2^p$.  We then compute $\nu(b+1)$ in $O(m)$ time. As $C(a,a)$ can be computed in $O(1)$ time for any $1\leq a\leq m$, we can compute all values $C(a,b)$ for $1\leq a\leq b\leq m$ in $O(m^3)$ time. Consider, for $1\leq\ell'\leq m'\leq m$, \[D(\ell',m')=\min_{\substack{\sigma'=(\sigma'_1,\ldots,\sigma'_{m'})\\S=(s_1,\ldots,s_{\ell'}) \text{ subsequence of $\sigma'$ }\\s_{\ell'}=\sigma_{m'}}}\sum_{i=1}^{\ell'}C(s_{i-1}+1,s_i),\]
    with $s_0:=0$, which corresponds to the optimal partitioning of the first $m'$ indices into $\ell'$ contiguous disjoint intervals under the cost function $C$. Computing the value $D(\ell,m)$ takes $O(m^2\ell)$ time via the recursive formula $D(\ell',m')=\min_{l\leq m'}D(\ell'-1,l) + C(l+1,m')$. The subsequence realizing $D(\ell,m)$ defines $\sigma^*$. For correctness, observe that the optimal simplification also yields a partition of the vertices of $\sigma$ into $\ell$ contiguous disjoint intervals. Let $T$ be the partition computed, and let $T\opt=([a_1,b_1],\ldots,[a_\ell,b_\ell])$ be an optimal partition of $[m]$, that realizes $\inf_{\sigma_{\ell}\in \curvespace{\ell}}\dtwp(\sigma_{\ell},\sigma)$. For a vertex $v_i$ in a fixed optimal simplification, let $\pi_i$ be a closest point among $\{\sigma_{a_i},\ldots,\sigma_{b_i}\}$. Then
    \begin{align*}
        \left(\sum_{[a,b]\in T}C(a,b)\right)^{1/p}&=\left(\sum_{[a,b]\in T}\min_{a\leq i\leq b}\sum_{a\leq j\leq b}\|\sigma_i-\sigma_j\|_2^p\right)^{1/p}\\
        &\leq \left(\sum_{[a,b]\in T\opt}\min_{a\leq i\leq b}\sum_{a\leq j\leq b}\|\sigma_i-\sigma_j\|_2^p\right)^{1/p}\\
        &\leq \left(\sum_{[a,b]\in T\opt}\sum_{a\leq j\leq b}\|\pi_i-\sigma_j\|_2^p\right)^{1/p}\\
        &\leq \left(\sum_{[a,b]\in T\opt}\sum_{a\leq j\leq b}(\|\pi_i-v_i\|_2+\|v_i-\sigma_j\|_2)^p\right)^{1/p}\\
        &\leq \left(\sum_{[a,b]\in T\opt}\sum_{a\leq j\leq b}(2\|v_i-\sigma_j\|_2)^p\right)^{1/p}\\
        &\leq 2\left(\sum_{[a,b]\in T\opt}\sum_{a\leq j\leq b}\|v_i-\sigma_j\|_2^p\right)^{1/p}=2\inf_{\sigma_{\ell}\in \curvespace{\ell}}\dtwp(\sigma_{\ell},\sigma).\qedhere
    \end{align*}
\end{proof}

\end{document}